\documentclass[twoside,11pt]{article}
\usepackage{jair, theapa, rawfonts, tikz, amsthm, subcaption, amsmath, algorithm, algpseudocode, multirow, placeins, url}

\usetikzlibrary{shapes.geometric}

\theoremstyle{plain}

\newtheorem{lemma}{Lemma}
\theoremstyle{definition}
\newtheorem{defn}{Definition}

\jairheading{1}{2020}{1-15}{6/91}{9/91}
\ShortHeadings{Optimal Any-Angle Pathfinding on a Sphere}
{Rospotniuk and Small}
\firstpageno{1}

\begin{document}
\title{Optimal Any-Angle Pathfinding on a Sphere}

\author{\name Volodymyr Rospotniuk \email volodymyr.rospotniuk@90poe.io \\
\name Rupert Small \email rupert.small@90poe.io \\
       \addr Ninety Percent of Everything, Portman House,\\
       2 Portman Street, London, United Kingdom, W1H 6DU}

\maketitle

\begin{abstract}
Pathfinding in Euclidean space is a common problem faced in robotics and computer games. For long-distance navigation on the surface of the earth or in outer space however, approximating the geometry as Euclidean can be insufficient for real-world applications such as the navigation of spacecraft, aeroplanes, drones and ships. This article describes an any-angle pathfinding algorithm for calculating the shortest path between point pairs over the surface of a sphere. Introducing several novel adaptations, it is shown that Anya as described by \cite{harabor} for Euclidean space can be extended to Spherical geometry. There, where the shortest-distance line between coordinates is defined instead by a great-circle path, the optimal solution is typically a curved line in Euclidean space. In addition the turning points for optimal paths in Spherical geometry are not necessarily corner points as they are in Euclidean space, as will be shown, making further substantial adaptations to Anya necessary.
\textit{Spherical Anya} returns the optimal path on the sphere, given these different properties of world maps defined in Spherical geometry. It preserves all primary benefits of Anya in Euclidean geometry, namely the \textit{Spherical Anya} algorithm always returns an optimal path on a sphere and does so entirely on-line, without any preprocessing or large memory overheads. Performance benchmarks are provided for several game maps including Starcraft and Warcraft III as well as for sea navigation on Earth using the NOAA bathymetric dataset. Always returning the shorter path compared with the Euclidean approximation yielded by Anya, \textit{Spherical Anya} is shown to be faster than Anya for the majority of sea routes and slower for Game Maps and Random Maps.
\end{abstract}

\section{Introduction}
\label{Introduction}
Pathfinding on a two-dimensional Euclidean grid is a well-studied mathematical problem with a wide range of practical applications, from logistics to gaming. A fundamental limitation of many early path-finding algorithms in this domain is their inability to travel along any angle on the grid due to the search path being constrained to 8 degrees of freedom: a left, right, up or down step to the nearest neighbouring node, or some pairwise combination of these resulting in a diagonal step.

Any-angle pathfinding refers to all pathfinding methodologies which remove this artificial constraint to the points of the grid by permitting movement in arbitrary directions, many of these being variations on A* \cite{hart}, a breakthrough graph search algorithm which uses a heuristic function to determine the most promising directions to explore the graph and thereby to reach the target point. These include Field D* \cite{ferguson}, Theta* \cite{nash}, Block A* \cite{yap}, Hierarchical Pathfinding A* (HPA*)\cite{botea}, and Accelerated A* \cite{sislak}. Amongst other strategies, these algorithms take various novel approaches to partitioning the search space or post-processing A* routes by removing nodes satisfying certain line-of-sight or collision avoidance criteria in order to selectively reduce the total route distance compared with A*. In so doing the resulting paths are shorter than the paths produced by graph-constrained approaches such as A* and so return routes which are closer to optimal. Still further modifications involve machine learning, for example whereby the measured and heuristic cost functions of A* are learned using a Recurrent Neural Network (RNN) \cite{wang} and myriad others which build not on A* but instead apply agent-based models with gradient descent optimising a loss function, such as Deep Q-learning \cite{panov} or those utilising LSTMs on a training set (for instance one generated using A*) such as \cite{sartoretti}\cite{mayur} or GAN-based approaches such as \cite{soboleva}. These Machine Learning approaches forfeit accuracy for speed and do not guarantee optimality.

Existing any-angle pathfinding algorithms for Euclidean space which guarantee optimality include Visibility Graphs \cite{perez}, Tangent Graphs \cite{liu}, Continuous Dijkstra based approaches \cite{mitchell}\cite{hershberger} and of course Anya \cite{harabor}. Anya capitalises on the fact that all turning points on an optimal route in Euclidean space will coincide with the corner points of obstacles, enabling the optimal path to be obtained via online calculation of root point--interval pairs. These are projected to form visibility cones, permitting Anya to search wide intervals of space in aggregate, rather than evaluating visibility between neighbouring point-pairs at each grid point of the world map. This mechanism of projecting into space using line-of-sight intervals is what makes Anya superior in speed to all known path-finding algorithms with an optimality guarantee and indeed, for some maps, faster than many without this guarantee \cite{uras2015b}.

For real-world applications demanding high accuracy over long-distance voyages it becomes imperative to extend the approach of Anya to spherical coordinates. This opens up an abundance of applications such as space, flight and sea voyage planning, all unsuited to being treated as an approximation of a route in Euclidean space. Transforming the world map representation with a Gnomonic projection whence all straight lines are great circle lines and vice versa, at first glance permits Anya to be implemented without modification; great-circle lines being straight lines in the Gnomonic representation of a world map. The reason this approach falls short of achieving optimality is discussed in Appendix \textit{A}. Instead, using Anya as a foundation it is shown that a new algorithm, \textit{Spherical Anya}, can be defined which inherits the beneficial properties of its progenitor and generalises the algorithm to routes in spherical geometry. To the extent possible, the fundamental concepts and definitions of Anya are preserved. The definition of a \textit{Grid} in Spherical Anya remains almost the same, as do the definitions of \textit{intervals}, \textit{intermediate points} and \textit{intersections}. \textit{Flat projections} also remain virtually identical to flat projections in Anya, whereas \textit{search nodes}, \textit{cone projections}, \textit{successors} and \textit{observability} in Spherical Anya must be defined differently due to the great-circle geodesic of Spherical geometry. Most significantly, the foundational theorem on which Anya depends -- that of all turning points on an optimal any-angle path being corner points -- no longer holds in Spherical geometry. A new type of non-observable cone projection must be introduced as a result, without which not all intervals in the Spherical world map would be explored. This ensures that the optimality, completeness and termination guarantees established for Anya extend to Spherical Anya.

\section{Spherical Anya Components and Definitions}
\label{anya}
Anya pathfinding on a flat Euclidean grid takes a distinctly different approach to A-star varieties of pathfinding which restrict travel to discrete pairwise-visible points. Instead, the approach taken by Anya capitalises on the fact that all turning points in optimal Euclidean paths are corner points \cite{mitchell} making the fundamental unit of abstraction an interval of points $[a, b]$, being either traversable or non-traversable, as opposed to a sequence of discrete nodes. It is in this way, by defining the world map width $W$ height $H$, in terms of traversable and non-traversable intervals and subsequently generating cone-- and flat-- projections along these, that Anya will select the optimal traversable path through world map intervals without manually inspecting all rational points contained in the Farey Sequence of order $n = min(W,H)$. The result is an optimal pathfinding algorithm in Euclidean space with superior speed.

To exhibit the algorithm defining Spherical Anya it is useful to revisit and transport the fundamental terminologies and components of Anya, these being used either directly or adapted for Spherical Anya, where the geodesic is a great-circle line. Wherever possible the foundational definitions of \cite{harabor}\cite{harabor2} are maintained and adaptations and extensions are introduced only as necessitated by the novel properties of Spherical geometry versus Euclidean geometry.

\begin{figure*}[h!]
\centering
\begin{tikzpicture}[scale=1]
\node at (0,0) {\includegraphics[angle=-90,origin=c,width=.4\linewidth, keepaspectratio]{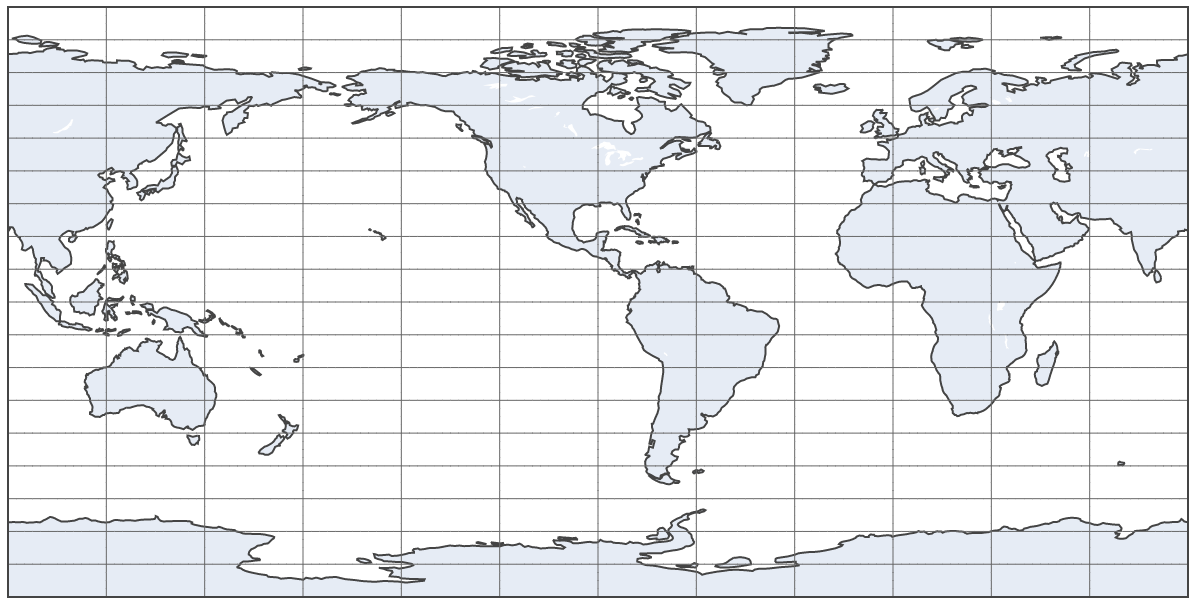}};
\node[rotate=90] at (-3.5,0) {$y$-axis (longitude)};
\node[] at (0,-6.5) {$x$-axis (latitude)};
\draw[->, very thick] (-3.2,-5.96) -- (3.4,-5.96);
\draw[->, very thick] (-2.95,-6.2) -- (-2.95,6.3);
\end{tikzpicture}

\caption{An Equirectangular projection map with $x$-axis defining lines of constant latitude and grid rows of the $y$-axis giving great-circle lines of constant longitude. Defining a constant $y$-value as lines of longitude ensures flat projections remain virtually identical in Spherical Anya as in Anya. Longitudinal values define \textit{rows} and pairs of latitudinal values, in combination with a row, define \textit{intervals}.}
\label{fig:equirectangular}
\end{figure*}
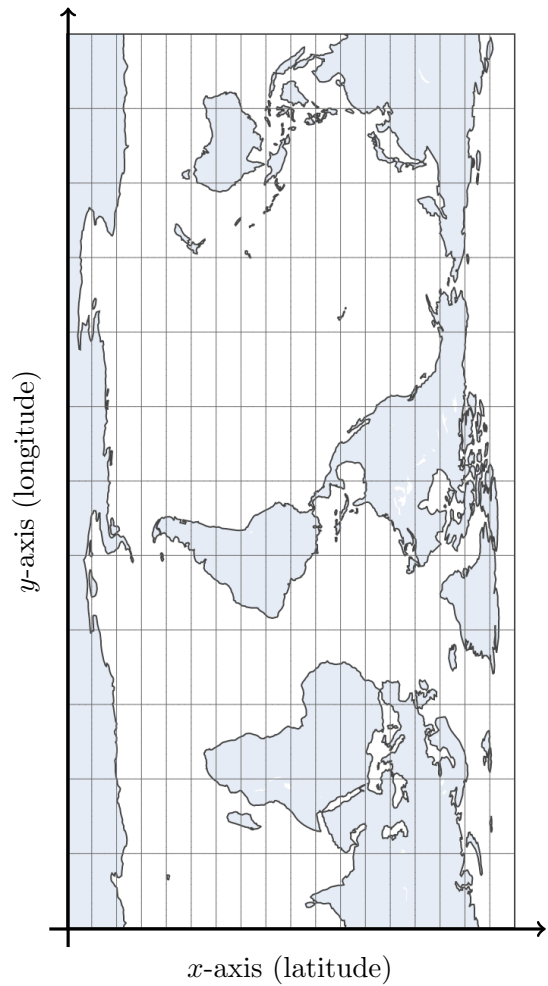
\begin{figure*}[h!]
\centering
\begin{subfigure}[t]{0.45\textwidth}
        \centering
        \begin{tikzpicture}[scale=1.4]
\draw[step=1cm,gray,very thin] (0,0) grid (4,3);
\filldraw[fill=black, draw=black] (0,1) rectangle (2,2);
\draw[red, very thick, dashed] (2.75,0) .. controls (1.6,1) and (1.6,2) .. (2.75,3);
\foreach \x in {0,1,2,3,4}
   \draw (\x cm,1pt) -- (\x cm,-1pt) node[anchor=north] {$\x$};
\foreach \y in {0,1,2,3}
    \draw (1pt,\y cm) -- (-1pt,\y cm) node[anchor=east] {$\y$};
\fill[red] (2, 2) circle[radius=1mm];
\fill[red] (2, 1) circle[radius=1mm];
\node[] at (2.2,1.2) {$a$};
\node[] at (2.2,2.2) {$b$};
\node[] at (0.2, -.5) {Southwards};
\node[] at (3.8, -.5) {Northwards};
\node[label={[align=left]Obstacle in\\ Southern\\ Hemisphere}, white] at (1.2, 2) {};
\draw[<->, very thick] (1, -.5) -- (3, -.5);
\end{tikzpicture}
 \caption{A great circle line between two points $a$ and $b$ of equal latitude in the southern hemisphere will always curve southwards, intersecting an obstacle with north-most edge upon this line. The symmetric claim holds likewise for a south-facing edge in the northern hemisphere.}\label{fig:great_circle_vs_vertical_barrier}
 \end{subfigure}%
 ~~~~~
 \begin{subfigure}[t]{0.45\textwidth}
        \centering
        \begin{tikzpicture}[scale=1.4]
\draw[step=1cm,gray,very thin] (0,0) grid (4,3);
\filldraw[fill=black, draw=black] (0,1) rectangle (2,2);
\draw[red, very thick, dashed] (1.25,0) .. controls (2.4,1) and (2.4,2) .. (1.25,3);
\foreach \x in {0,1,2,3,4}
   \draw (\x cm,1pt) -- (\x cm,-1pt) node[anchor=north] {$\x$};
\foreach \y in {0,1,2,3}
    \draw (1pt,\y cm) -- (-1pt,\y cm) node[anchor=east] {$\y$};
\fill[red] (2, 2) circle[radius=1mm];
\fill[red] (2, 1) circle[radius=1mm];
\node[] at (2.2,1.2) {$a$};
\node[] at (2.2,2.2) {$b$};
\node[] at (0.2, -.5) {Southwards};
\node[] at (3.8, -.5) {Northwards};
\node[label={[align=left]Obstacle in\\ Northern\\ Hemisphere}, white] at (1.2, 2) {};
\draw[<->, very thick] (1, -.5) -- (3, -.5);
\end{tikzpicture}
 \caption{A great circle line between two points $a$ and $b$ of equal latitude in the northern hemisphere will always curve northwards, avoiding an obstacle with north-most edge upon this line. The symmetric claim holds likewise for a south-facing edge in the southern hemisphere.
 }
 \end{subfigure}%
  \caption{Great-circle lines passing through two vertical corners of an obstacle will either intersect only with points internal to the bulk of the obstacle or only with points external to it, with the exception of the corner points $a$ and $b$ themselves.}\label{fig:great_circle_vs_vertical_barrier}
\label{fig:great_circle_at_edges}
\end{figure*}
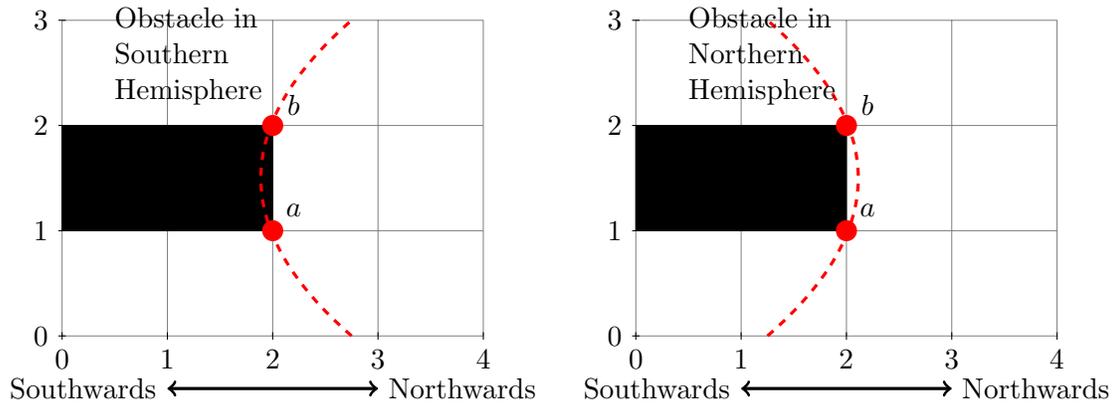

For Spherical Anya the grid remains a planar subdivision of $W \times H$ cells. Rows are lines of constant longitude and vertical edges are lines of constant latitude, as shown in Figure \ref{fig:equirectangular}. This means each cell in Spherical Anya is neither square in the Euclidean sense nor the Spherical sense; the vertical edges of the cell are not great-circle lines, albeit the horizontal edges of the cell \textit{are}, and therefore constitute the shortest path between points on the same horizontal line, or row. Rows thus defined, obstacle edges -- facing either the north or south poles of the sphere -- are not great-circle lines with the exception of edges on the equator. Hence a great-circle line passing through the vertically aligned corner points of an obstacle defined on the grid will in general either pass exclusively within the non-traversable bulk of the obstacle or exclusively external to it, with the exception of the corner points themselves, as illustrated in Figure \ref{fig:great_circle_at_edges}.

Of particular note is that the existence of a line of constant latitude between two points, in addition passing adjacent only to traversable cells, does not guarantee that these two points are \textit{visible}, as illustrated for example in Figure \ref{fig:great_circle_at_edges}(a). Instead the great-circle line, rather than a straight line in Euclidean space, underlies the definition of visibility in Spherical geometry.

\begin{defn}
	\textit{Two points are said to be visible if they can be connected by a great-circle line which does not pass through an obstacle or the intersection formed by two diagonally-adjacent obstacles.}
\end{defn}

\noindent By fixing rows as lines of constant longitude, the definition of a grid interval in Spherical Anya can remain virtually identical as for Anya.

\begin{defn}
	\textit{A grid interval $I$ is a set of contiguous and pairwise visible points drawn from any discrete row of constant longitude on the grid. Each interval is defined in terms of its endpoints $a$ and $b$. With the possible exception of $a$ and $b$, each interval contains only intermediate non-corner points.}
\end{defn}

\noindent The grid, intervals and visibility thus defined, the next building block of Anya is the \textit{search node}. For Spherical Anya it becomes necessary to define not one but two categories of search node, \textit{standard} and \textit{adjoint}, the curvature of Spherical geometry making the second category necessary to guarantee completeness of the search space and thereby optimality of the resulting path.

\begin{defn}
\textit{A standard search node $(I, r)$ is a tuple where $r \notin I$ is a point called the root and $I$ is an interval such that each point $p \in I$ is visible from r.}
\label{defn:standard_node}
\end{defn}

\noindent A standard node here is much the same as a node in Anya albeit with the altered definition of visibility, being defined by the existence of a traversable great-circle path between the root $r$ and all points in the interval $I$ as illustrated in Figure \ref{fig:cone_projection}, with $I := [a, b]$. The interval may be closed, as in this example, or alternatively open, or half-open. The kind of interval used will depend on later considerations (see Definitions \ref{defn:standard_successor} and \ref{defn:adjoint_successor}). The second category of search node, which exists only in Spherical Anya, arises by virtue of the vertical line of an obstacle edge generally \textit{not} being the shortest path between the two corner points on the vertical edge, as illustrated in Figure \ref{fig:intermediate_cone_projection}.

\begin{figure*}[h]
\centering
\begin{tikzpicture}[scale=1]
\draw[step=1cm,gray,very thin] (0,0) grid (4,4);
\draw[line width=1.5mm, red] (0,4) -- (4,4);
\draw[red, very thick, dashed] (1,0) .. controls (1,1) and (.5,3) .. (0,4);
\draw[red, very thick, dashed] (1,0) .. controls (2.4,.5) and (3.6, 3) .. (4,4);
\foreach \x in {0,1,2,3,4}
   \draw (\x cm,1pt) -- (\x cm,-1pt) node[anchor=north] {$\x$};
\foreach \y in {0,1,2,3,4}
    \draw (1pt,\y cm) -- (-1pt,\y cm) node[anchor=east] {$\y$};
\fill[red] (1, 0) circle[radius=1.3mm];
\node[] at (1.3,0.4) {$r$};
\node[] at (0,4.5) {$a$};
\node[] at (4,4.5) {$b$};
\end{tikzpicture}
\caption{The standard node in Spherical Anya $[a, b], r$ is an interval $I := [a, b]$ and a root point $r$ with $r \notin [a, b]$, where a traversible great-circle path exists between all points $p \in I$. Standard nodes form the first primary abstraction of the search space, all successors along a path from source to target destination consisting of a sequence of root points which are either root points of a standard node (definition \ref{defn:standard_node}) or an adjoint node (definition \ref{defn:intermediate_node}).}
\label{fig:cone_projection}
\end{figure*}
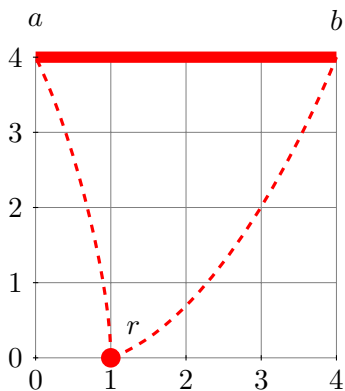\label{cone_projection}

\begin{defn}
\textit{An adjoint search node $(I, r)$ is a tuple where $r \notin I$ is a point called the root, which is located on the vertical edge of an obstacle, and $I$ is a traversable interval such that each point $p \in I$ is invisible from r and furthermore the great-circle line from $r$ to $p$ intersects the vertical edge twice, and only the vertical edge, including once at $r$.}
\label{defn:intermediate_node}
\end{defn}

\noindent This second kind of search node, illustrated in Figure \ref{fig:intermediate_cone_projection}, is unique to Spherical Anya and is not required in Euclidean space where vertical edges describe straight lines and the turning points of an optimal route will always be corner points. Neither is true in Spherical geometry (see Section \ref{section:optimality}), where by convention the rows of the world map grid have been fixed as great-circle lines of constant longitude and the vertical edges of obstacles are lines of constant latitude. 
 
\begin{figure*}[h]
\centering
\begin{tikzpicture}[scale=1.4]
\draw[step=1cm,gray,very thin] (0,0) grid (4,3);
\filldraw[fill=black, draw=black] (0,0) rectangle (2,3);
\draw[red, line width=1.5mm] (2,3) -- (4,3);
\draw[red, very thick, dashed] (2,0) .. controls (1.5,1) and (1.5,2) .. (2,3);
\draw[red, very thick, dashed] (2,0) .. controls (2,1.5) and (3,2.5) .. (4,3);
\foreach \x in {0,1,2,3,4}
   \draw (\x cm,1pt) -- (\x cm,-1pt) node[anchor=north] {$\x$};
\foreach \y in {0,1,2,3}
    \draw (1pt,\y cm) -- (-1pt,\y cm) node[anchor=east] {$\y$};
\fill[red] (2, 0) circle[radius=1mm];
\node[] at (2.2,0.2) {$r$};
\node[] at (2,3.3) {$a$};
\node[] at (4,3.3) {$b$};
\end{tikzpicture}
        \caption{An adjoint node exists only in Spherical Anya due to the relative curvature between the vertical constant-latitude edge of obstacles and a great-circle line drawn between the two vertically aligned corner points of such an edge. An adjoint node is a node $[a, b), r$ where all cells adjacent to $[a, b)$ are traversable and furthermore all great circle lines through the root $r$ and any point $p \in [a, b)$ intersect the vertical edge of the obstacle twice.  The opposite scenario exists in Spherical geometry too, where no great-circle line between the vertically aligned corners of an obstacle edge intersects the edge, excepting at the corner points.}
    \label{fig:intermediate_cone_projection}
\end{figure*}
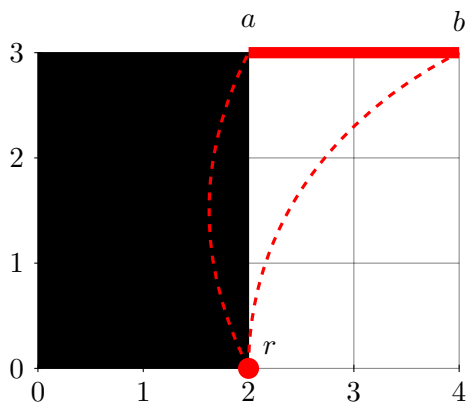

The existence of adjoint nodes in Spherical geometry marks the most significant departure between the underlying mathematical structure of Spherical Anya compared to Anya. Thus far the differences have been limited; the path of shortest distance merely being substituted for the great-circle path of Spherical geometry, the orientation of the grid fixed explicitly to ensure rows are shortest-path lines of longitude, and the definition of visibility similarly adapted to fit the new context. However, with the addition of adjoint nodes the divergence becomes more fundamental. This divergence is formalised in Section \ref{section:optimality}, after which it finally becomes practical to provide the concluding definitions required to implement Spherical Anya, those of \textit{successors} and \textit{projections}.

\section{Optimality}\label{section:optimality}

\begin{lemma}
	Not all turning points on an optimal path in spherical geometry are corner points.
\label{lemma:turning_points}
\end{lemma}
\begin{proof}
	By construction. It is assumed the obstacle edge is located in the northern hemisphere and facing southwards as illustrated in Figure \ref{fig:ersatz_vs_optimal_distance_metric}. All results follow identically for the southern hemisphere and without loss of generality. For longitude, latitude $(\theta, \phi)$ and Gnomonic projection centred on the north pole $(0, \frac{\pi}{2})$ we have projected coordinates

\begin{align}\label{eq:gnomonic_mapping}
x' &= \frac{cos(\phi)sin(\theta)}{sin(\phi)}\nonumber\\
y' &= -\frac{cos(\phi)cos(\theta)}{sin(\phi)}
\end{align}

\noindent As the shortest path between any pair of points under the Gnomonic projection is a straight line we can define the last point on the obstacle which remains visible to $\mathbf{p}$ as $\mathbf{x}$, defined by coordinates $(\theta_1, \phi_1)$. Target point $\mathbf{p}$ takes coordinates $(\theta_2, \phi_2)$. The last point on the obstacle with direct line-of-sight to $\mathbf{p}$ must satisfy $\mathbf{x} \cdot (\mathbf{p}-\mathbf{x}) = 0$. Plugging into (\ref{eq:gnomonic_mapping}) yields $\mathbf{x}\cdot \mathbf{p} = \frac{cos(\phi_1)cos(\phi_2)}{sin(\phi_1)sin(\phi_2)} cos(\theta_2 - \theta_1)$ and $||\mathbf{x}||^2 = \frac{cos^2(\phi_1)}{sin^2(\phi_1)}$. Solving for $\theta_1$ gives

\begin{equation}\label{eq:departure_angle}
	\theta_1 = \theta_2 - cos^{-1} \left [ \frac{tan(\phi_2)}{tan(\phi_1)} \right ] 
\end{equation}

\noindent This defines the point of departure from the obstacle edge. Now we show that the result which follows from applying string-pulling holds, namely that the shortest path from the root point $\mathbf{r}$ to $\mathbf{p}$ consists of the line segment $L(\mathbf{r}, \mathbf{x})$ defining the boundary of the obstacle between $\mathbf{r}$ and $\mathbf{x}$, plus the great circle line $C(\mathbf{x}, \mathbf{p})$ between $\mathbf{x}$ and $\mathbf{p}$. The central angle between $\mathbf{p}$ and any point $\mathbf{x}'$ on the line $\theta=\theta_1$ is given by

\begin{equation}
	\Delta\sigma = cos^{-1}\left(sin(\phi')sin(\phi_2) + cos(\phi')cos(\phi_2)cos(\theta_2 - \theta_1)\right)
\end{equation}

\noindent From identity $\frac{\partial}{\partial x}cos^{-1}(x) = \frac{-1}{\sqrt{1-x^2}}$ we have that as $\phi'$ decreases away from $\phi_1$ (recall the equator is at $\phi=0$) the central angle $\Delta\sigma$ and therefore the great-circle distance changes by $-\frac{\partial}{\partial \phi'}$ which has the sign of $\mathrm{sgn}\left[cos(\phi')sin(\phi_2) - sin(\phi') cos(\phi_2)cos(\theta_2 - \theta_1)\right].$ Substituting in $\theta_1$ from (\ref{eq:departure_angle}) gives $\mathrm{sgn}\left[tan(\phi_1) - tan(\phi')\right]$ and since $\phi' < \phi_1$ this is always positive

\begin{equation}
\mathrm{sgn}\left[-\frac{\partial}{\partial \phi'}\Delta\sigma\right] = \mathrm{sgn}\left[tan(\phi_1) - tan(\phi')\right]  = 1.
\end{equation}

\noindent It follows that the length of this segment always increases with decreasing $\phi'$ and we can conclude that not only is $C(\mathbf{x}, \mathbf{p})$ the shortest traversable path between $\mathbf{x}$ and $\mathbf{p}$ but it is also the shortest traversable path between $\mathbf{p}$ and any point $\mathbf{x}'$ on $\theta=\theta_1$ conditioned on $\phi' < \phi_1$. Finally, since $R\cdot cos\phi_1$ is the lower bound on the radius for all paths between longitudes $\theta_0$ and $\theta_1$ in the northern hemisphere it follows that the shortest distance from $r$ to the line defined by $\theta = \theta_1$ is $R \cdot cos\phi_1 (\theta_1 - \theta_0)$, which is the path defining the obstacle boundary between those points. Hence the turning point for the optimal path between $\mathbf{p}$ and $\mathbf{r}$ is $\mathbf{x}$, as would have been concluded by string-pulling in the plane under a Gnomonic projection.
\end{proof}

\noindent For completeness, the practical consequence of Lemma \ref{lemma:turning_points} within the grid as defined for Spherical Anya is also provided explicitly as an additional Lemma.

\begin{lemma}
	All turning points on an optimal path in Spherical geometry are corner points or lie within the interval defined by two vertically aligned corner points of an obstacle.
\label{lemma:turning_points_lemma}
\end{lemma}
\begin{proof}
	The statement is equivalent to the claim that no turning point lies within the horizontal interval of a barrier. True by construction; horizontal lines are great-circle lines.
\end{proof}

\begin{figure*}[h]
\begin{subfigure}[t]{0.45\textwidth}
\centering
\begin{tikzpicture}[scale=1.5]
\draw[step=1cm,gray,very thin] (0,0) grid (4,3);
\filldraw[fill=black, draw=black] (0,0) rectangle (2,3);

\draw[red, line width=1.5mm] (2,3) -- (4,3);
\draw[white, very thick, dashed] (2,0) .. controls (1.5,1) and (1.5,2) .. (2,3);
\draw[red, very thick, dashed] (2,0) .. controls (2,1.5) and (3,2.5) .. (4,3);
\draw[red, very thick, dashed] (2,1) .. controls (2,1.8) and (2.3,2.5) .. (3,3);

\draw (2 cm,1pt) -- (2 cm,-1pt) node[anchor=north] {$\phi_1$};
\draw (3 cm,1pt) -- (3 cm,-1pt) node[anchor=north] {$\phi_2$};
\draw (1pt,0 cm) -- (-1pt,0 cm) node[anchor=east] {$\theta_0$};
\draw (1pt,1 cm) -- (-1pt,1 cm) node[anchor=east] {$\theta_1$};
\draw (1pt,3 cm) -- (-1pt,3 cm) node[anchor=east] {$\theta_2$};

\fill[white] (2, 0) circle[radius=1mm];
\fill[red] (2, 0) circle[radius=.8mm];
\fill[white] (2, 1) circle[radius=.8mm];
\fill[red] (2, 1) circle[radius=.6mm];
\fill[white] (3, 3) circle[radius=1.3mm];
\fill[red] (3, 3) circle[radius=1mm];
\node[] at (2.2,0.2) {$r$};
\node[] at (2,3.3) {$a$};
\node[] at (4,3.3) {$b$};
\node[white] at (1.8,1) {$x$};
\node[] at (3,3.28) {$p$};
\end{tikzpicture}
\caption{Adjoint node $[a, b), r$ is constructed such that $[a, b)$ is traversable and for any $p \in [a, b)$ the optimal route between $r$ and $p$ will consist of a constant-latitude segment $L(r, x)$ and a great-circle segment $C(x, p)$. Hence $x$ is the turning point of the optimal path between the root $r$ and $p$, which is the basis of Lemma's \ref{lemma:turning_points} and \ref{lemma:turning_points_lemma}.}
\end{subfigure}%
~~~~~~
\begin{subfigure}[t]{0.45\textwidth}
        \centering
\begin{tikzpicture}[scale=1.5]

\draw[black, line width=15mm] (3,0) arc (0:42.7:2.85);
\draw[gray, very thick, dotted] (3.5,0) arc (0:60:3.3);
\draw[gray, very thick, dotted] (4,0) arc (0:55:3.75);
\draw[red, very thick, dashed] (3,2.6) -- (3.39,.88);
\draw[red, very thick, dashed] (3.5,0) -- (3.5,3);
\draw[gray, thick, dotted] (0,0) -- (4,3.4);
\draw[gray, thick, dotted] (0,0) -- (4.4,0);
\draw[gray, thick, dotted] (0,0) -- (4.3,1.15);
\fill[white] (3.5, 0) circle[radius=1.1mm];
\fill[red] (3.5, 0) circle[radius=.8mm];
\draw[red, line width=1.5mm] (2.62,2.22) -- (3.5,3);

\fill[gray] (0,0) circle[radius=.5mm];
\fill[white] (3, 2.55) circle[radius=1.1mm];
\fill[red] (3, 2.55) circle[radius=.8mm];
\fill[white] (3.39,.88) circle[radius=.9mm];
\fill[red] (3.39,.88) circle[radius=.65mm];
\node[] at (3.65,.2) {$r$};
\node[] at (2.5,2.4) {$a$};
\node[] at (3.5,3.2) {$b$};
\node[white] at (3.14,.88) {$x$};
\node[] at (2.9,2.75) {$p$};
\node[] at (2,3) {$\phi_1$};
\node[] at (2.4,3.2) {$\phi_2$};
\node[] at (4.3,1.2) {$\theta_1$};
\node[] at (4.5,0) {$\theta_0$};
\node[] at (4,3.5) {$\theta_2$};
\node[] at (.1,-.2) {$(0, \pi/2)$};
\node[] at (3, -.3) {decreasing $\phi$};
\draw[->, gray, very thick] (3.8, -.3) -- (4.4, -.3);

\end{tikzpicture}
\caption{Gnomonic projection of the obstacle of Figure \ref{fig:ersatz_vs_optimal_distance_metric}(a). Great-circles are straight lines and without loss of generality the centre is taken as the north pole $(0, \pi/2)$. Circles centred on the origin are lines of constant latitude, with $\phi$ decreasing to zero at the equator as the radius increases.}
\end{subfigure} 
\caption{}
\label{fig:ersatz_vs_optimal_distance_metric}
\end{figure*}

\noindent Lemma's \ref{lemma:turning_points} and \ref{lemma:turning_points_lemma} provide the foundation on which to present the final components required to implement the Spherical Anya algorithm, that of a \textit{successor}, which is again adapted as minimally as possible from that of a successor in Anya, and three categories of projection; \textit{flat projections}, \textit{standard cone projections} and, uniquely to Spherical Anya, \textit{adjoint cone projections}.

Whereas the optimal path from source to target in Anya is a sequence of pairwise visible root points, in Spherical Anya the optimal path is a sequence of \textit{turning points} which, with the exception of turning points lying on the vertical edge of an obstacle, are also pairwise-visible root points, by Lemma's \ref{lemma:turning_points} and \ref{lemma:turning_points_lemma}. In spite of this and the significant adaptations made, the algorithm for calculating this sequence of optimal turning points is analogous to Anya, where successors are constructed at each node, the mechanism by which this occurs underlying the guarantee of completeness, optimality and termination identically in Spherical Anya as in Anya.

\section{Successors and Projections}
Successors in Spherical Anya are separated into two categories, each according to the kind of Anya node associated with it; standard or adjoint.

\begin{defn}
\textit{A standard successor of a search node $(I,r)$ is a search node $(I',r')$ such that
\begin{enumerate}
	\item $(I',r')$ is a standard node.
	\item For all points $p' \in I'$ there exists a point $p \in I$ such that the great-circle segments $C(r, p, p')$ form a locally optimal route from $r$ to $p'$.
	\item $r'$ is the unique point closest to $p'$ shared by all great-circle segments $C(r, p, p')$.
	\item $I'$ is maximal given the points above and the definition of a standard node.
\end{enumerate}
}
\label{defn:standard_successor}
\end{defn}
\noindent This is the analog of a successor as defined in Anya. Due to the existence of adjoint nodes in Spherical geometry, it is also necessary to define an \textit{adjoint successor}. For all intents and purposes these are to be treated as a standard successor in the implementation of the Spherical Anya algorithm; the fact that no point in the node interval is visible from the root point is relevant only insofar as local optimality of the path from adjoint node root $r$ through $x$ to any point in $I$ is in those cases accorded via Lemma 1, as is the corresponding addition to the route $g-$value. Hence

\begin{align}
f\mathrm{-value} &= g\mathrm{-value} + \mathrm{Anya~heuristic	}\nonumber\\
&= |C(s, r_1, ..., r_{current})| + |C(r_{current}, t)|
\end{align}

\noindent is throughout given by the sum of the \textit{shortest-path} line segments formed by the sequence of \textit{root} points $r_1, ... , r_{current}$ within each node from the source $s$, plus the Anya heuristic, being the great-circle distance between the current root and the target $t$. The shortest-path line segments $C$ are great-circles in all cases except for adjoint nodes, where the path and therefore the corresponding addition to the $g-$value is described by Lemma 1.

\begin{defn}
\textit{An adjoint successor of a search node $(I,r) = ((a, b), r)$ is a search node $(I',r') = ((a', b'), r')$ such that
\begin{enumerate}
	\item $(I', r')$ is an adjoint node.
	\item $I'$ is maximal given the definition of an adjoint node.
	\item Either $r' = a$ or $r' = b$.
\end{enumerate}
}
\label{defn:adjoint_successor}
\end{defn}

The mechanism by which successors are generated is via \textit{projections}, analogously to Anya, these being {(i)} observable (visible) cone projections, {(ii)} non-observable (invisible) standard cone projections, {(iii)} non-observable adjoint cone projections or {(iv)} flat projections. Some of these projections may themselves be broken into subcategories dependent on context. It is useful to give an overview of these in Anya, as the terminology and concepts are foundational to both algorithms.

\subsection{Projections in Anya}
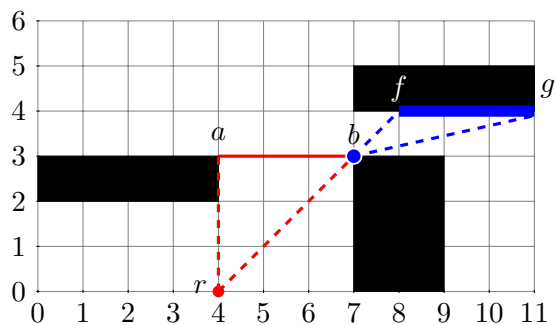
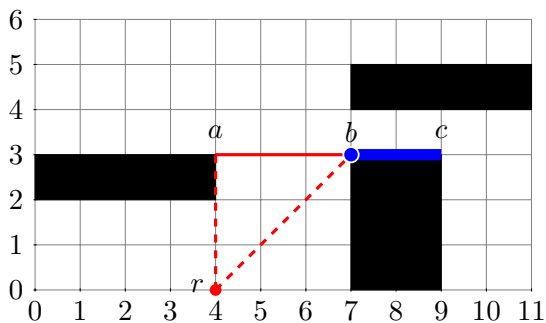
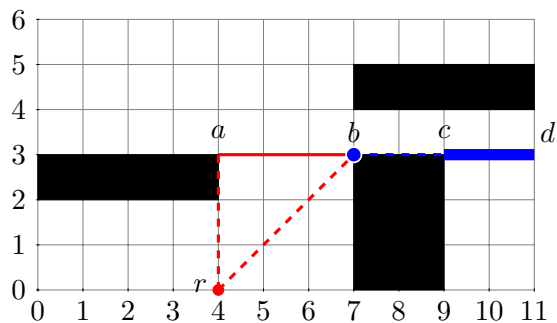
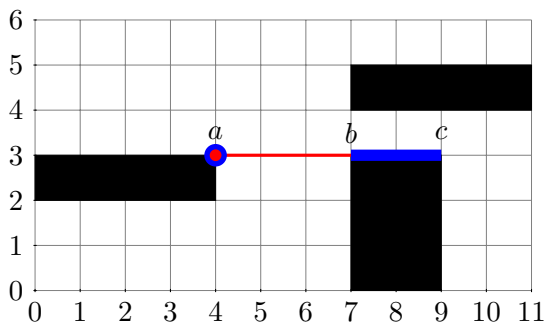
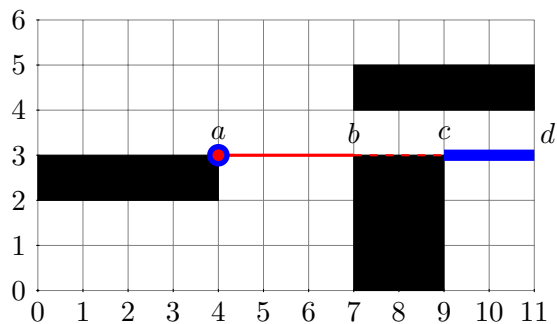
\begin{figure*}[t!]
    \centering
    \begin{subfigure}[t]{0.5\textwidth}
        \centering
\begin{tikzpicture}[scale=.6]
\draw[step=1cm,gray,very thin] (0,0) grid (11,6);
\filldraw[fill=black, draw=black] (0,2) rectangle (4,3);
\filldraw[fill=black, draw=black] (7,3) rectangle (9,0);
\filldraw[fill=black, draw=black] (7,5) rectangle (11,4);
\draw[red, very thick] (4,3) -- (7,3);
\draw[red, very thick, dashed] (4,0) -- (4,4);
\draw[red, very thick, dashed] (4,0) -- (8,4);
\draw[line width=1.5mm, blue] (4,4) -- (8,4);
\fill[blue] (4, 0) circle[radius=2.5mm];
\foreach \x in {0,1,2,3,4,5,6,7,8,9,10,11}
   \draw (\x cm,1pt) -- (\x cm,-1pt) node[anchor=north] {$\x$};
\foreach \y in {0,1,2,3,4,5,6}
    \draw (1pt,\y cm) -- (-1pt,\y cm) node[anchor=east] {$\y$};
\fill[red] (4, 0) circle[radius=1.3mm];
\node[] at (3.6,0.1) {$r$};
\node[] at (4,3.5) {$a$};
\node[] at (7,3.5) {$b$};
\node[] at (4,4.5) {$e$};
\node[white] at (8,4.5) {$f$};
\end{tikzpicture}
        \caption{An observable cone projection. In this instance the visibility cone is projected up one line to yield successor node $[e, f], r$.}
    \end{subfigure}%
    ~~~~~~~~ 
    \begin{subfigure}[t]{0.5\textwidth}
        \centering
\begin{tikzpicture}[scale=.6]
\draw[step=1cm,gray,very thin] (0,0) grid (11,6);
\filldraw[fill=black, draw=black] (0,2) rectangle (4,3);
\filldraw[fill=black, draw=black] (7,3) rectangle (9,0);
\filldraw[fill=black, draw=black] (7,5) rectangle (11,4);
\draw[red, very thick] (4,3) -- (7,3);
\draw[red, very thick, dashed] (4,0) -- (4,3);
\draw[red, very thick, dashed] (4,0) -- (7,3);
\draw[blue, very thick, dashed] (7,3) -- (8,4);
\draw[blue, very thick, dashed] (7,3) -- (11,3.9);
\draw[line width=1.5mm, blue] (8,4) -- (11,4);
\fill[white] (7, 3) circle[radius=1.9mm];
\fill[blue] (7, 3) circle[radius=1.5mm];
\foreach \x in {0,1,2,3,4,5,6,7,8,9,10,11}
   \draw (\x cm,1pt) -- (\x cm,-1pt) node[anchor=north] {$\x$};
\foreach \y in {0,1,2,3,4,5,6}
    \draw (1pt,\y cm) -- (-1pt,\y cm) node[anchor=east] {$\y$};

\fill[red] (4, 0) circle[radius=1.3mm];
\node[] at (3.6,0.1) {$r$};
\node[] at (4,3.5) {$a$};
\node[] at (7,3.5) {$b$};
\node[] at (11.3,4.5) {$g$};
\node[white] at (8,4.5) {$f$};
\end{tikzpicture}
        \caption{A non-observable cone projection. In this case, $b$ being a successor root node, the cone projection of $[a, b], r$ is found to be $(f, g], b$. No point within $(f, g]$ is visible from $r$.}
    \end{subfigure}\newline\newline
    \begin{subfigure}[t]{0.5\textwidth}
        \centering
\begin{tikzpicture}[scale=.6]
\draw[step=1cm,gray,very thin] (0,0) grid (11,6);
\filldraw[fill=black, draw=black] (0,2) rectangle (4,3);
\filldraw[fill=black, draw=black] (7,3) rectangle (9,0);
\filldraw[fill=black, draw=black] (7,5) rectangle (11,4);
\draw[red, very thick] (4,3) -- (7,3);
\draw[red, very thick, dashed] (4,0) -- (4,3);
\draw[red, very thick, dashed] (4,0) -- (7,3);
\draw[line width=1.5mm, blue] (7,3) -- (9,3);
\fill[white] (7, 3) circle[radius=1.9mm];
\fill[blue] (7, 3) circle[radius=1.5mm];
\foreach \x in {0,1,2,3,4,5,6,7,8,9,10,11}
   \draw (\x cm,1pt) -- (\x cm,-1pt) node[anchor=north] {$\x$};
\foreach \y in {0,1,2,3,4,5,6}
    \draw (1pt,\y cm) -- (-1pt,\y cm) node[anchor=east] {$\y$};
\fill[red] (4, 0) circle[radius=1.3mm];
\node[] at (3.6,0.1) {$r$};
\node[] at (4,3.5) {$a$};
\node[] at (7,3.5) {$b$};
\node[] at (9,3.5) {$c$};
\end{tikzpicture}
        \caption{An intermediate non-observable cone projection. The successor node $(b, c], b$ to $[a,b], r$ is said to be intermediate because it itself has only one successor, namely $(c, d], c$ in Figure \ref{fig:anya_projections}(d).}
    \end{subfigure}%
    ~~~~~~~~ 
    \begin{subfigure}[t]{0.5\textwidth}
        \centering
\begin{tikzpicture}[scale=.6]
\draw[step=1cm,gray,very thin] (0,0) grid (11,6);
\filldraw[fill=black, draw=black] (0,2) rectangle (4,3);
\filldraw[fill=black, draw=black] (7,3) rectangle (9,0);
\filldraw[fill=black, draw=black] (7,5) rectangle (11,4);
\draw[red, very thick] (4,3) -- (7,3);
\draw[red, very thick, dashed] (4,0) -- (4,3);
\draw[red, very thick, dashed] (4,0) -- (7,3);
\draw[blue, very thick, dashed] (7,3) -- (9,3);
\draw[line width=1.5mm, blue] (9,3) -- (11,3);
\fill[white] (7, 3) circle[radius=1.9mm];
\fill[blue] (7, 3) circle[radius=1.5mm];
\foreach \x in {0,1,2,3,4,5,6,7,8,9,10,11}
   \draw (\x cm,1pt) -- (\x cm,-1pt) node[anchor=north] {$\x$};
\foreach \y in {0,1,2,3,4,5,6}
    \draw (1pt,\y cm) -- (-1pt,\y cm) node[anchor=east] {$\y$};
\fill[red] (4, 0) circle[radius=1.3mm];
\node[] at (3.6,0.1) {$r$};
\node[] at (4,3.5) {$a$};
\node[] at (7,3.5) {$b$};
\node[] at (9,3.5) {$c$};
\node[] at (11.3,3.5) {$d$};
\end{tikzpicture}
        \caption{A non-observable cone projection. In this illustration the non-observable successor $(c, d], b$ is non-intermediate because it may have multiple successors, including successors generated via cone projections.}
    \end{subfigure}
    \begin{subfigure}[t]{0.5\textwidth}
        \centering
\begin{tikzpicture}[scale=.6]
\draw[step=1cm,gray,very thin] (0,0) grid (11,6);
\filldraw[fill=black, draw=black] (0,2) rectangle (4,3);
\filldraw[fill=black, draw=black] (7,3) rectangle (9,0);
\filldraw[fill=black, draw=black] (7,5) rectangle (11,4);
\draw[red, very thick] (4,3) -- (7,3);
\draw[line width=1.5mm, blue] (7,3) -- (9,3);
\foreach \x in {0,1,2,3,4,5,6,7,8,9,10,11}
   \draw (\x cm,1pt) -- (\x cm,-1pt) node[anchor=north] {$\x$};
\foreach \y in {0,1,2,3,4,5,6}
    \draw (1pt,\y cm) -- (-1pt,\y cm) node[anchor=east] {$\y$};
\fill[blue] (4, 3) circle[radius=2.5mm];
\fill[red] (4, 3) circle[radius=1.3mm];
\node[] at (4,3.5) {$a$};
\node[] at (7,3.5) {$b$};
\node[] at (9,3.5) {$c$};
\end{tikzpicture}
        \caption{An intermediate flat projection. Here the search node $(a, b], a$ has been projected horizontally along the same row, generating successor $(b, c], a$, an intermediate successor as per Figure \ref{fig:anya_projections}(c).}
    \end{subfigure}%
    ~~~~~~~~ 
    \begin{subfigure}[t]{0.5\textwidth}
        \centering
\begin{tikzpicture}[scale=.6]
\draw[step=1cm,gray,very thin] (0,0) grid (11,6);
\filldraw[fill=black, draw=black] (0,2) rectangle (4,3);
\filldraw[fill=black, draw=black] (7,3) rectangle (9,0);
\filldraw[fill=black, draw=black] (7,5) rectangle (11,4);
\draw[red, very thick] (4,3) -- (7,3);
\draw[red, thick, dashed] (7,3) -- (9,3);
\draw[line width=1.5mm, blue] (9,3) -- (11,3);

\foreach \x in {0,1,2,3,4,5,6,7,8,9,10,11}
   \draw (\x cm,1pt) -- (\x cm,-1pt) node[anchor=north] {$\x$};
\foreach \y in {0,1,2,3,4,5,6}
    \draw (1pt,\y cm) -- (-1pt,\y cm) node[anchor=east] {$\y$};
\fill[blue] (4, 3) circle[radius=2.5mm];
\fill[red] (4, 3) circle[radius=1.3mm];
\node[] at (4,3.5) {$a$};
\node[] at (7,3.5) {$b$};
\node[] at (9,3.5) {$c$};
\node[] at (11.3,3.5) {$d$};
\end{tikzpicture}
        \caption{A non-intermediate flat projection. Here the search node $(a, b], a$ has been projected horizontally along the same row, generating successor $(c, d], a$. This is a non-intermediate successor as per Figure \ref{fig:anya_projections}(d).}
    \end{subfigure}    
    \caption{Anya projections on a Euclidean world map}
    \label{fig:anya_projections}
\end{figure*}

\begin{figure*}[t!]
    \centering
    \begin{subfigure}[t]{0.5\textwidth}
        \centering
\begin{tikzpicture}[scale=.6]
\draw[step=1cm,gray,very thin] (0,0) grid (11,6);
\filldraw[fill=black, draw=black] (0,2) rectangle (4,3);
\filldraw[fill=black, draw=black] (7,3) rectangle (9,0);
\filldraw[fill=black, draw=black] (7,5) rectangle (11,4);
\draw[red, very thick] (4,3) -- (7,3);
\draw[red, very thick, dashed] (4,0) .. controls (4.5,2) and (4.1,3) .. (3.7,4);
\draw[red, very thick, dashed] (4,0) .. controls (5,.5) and (7,2.6) .. (7.4,4);
\draw[line width=1.5mm, blue] (3.7,4) -- (7.4,4);
\fill[blue] (4, 0) circle[radius=2.5mm];
\foreach \x in {0,1,2,3,4,5,6,7,8,9,10,11}
   \draw (\x cm,1pt) -- (\x cm,-1pt) node[anchor=north] {$\x$};
\foreach \y in {0,1,2,3,4,5,6}
    \draw (1pt,\y cm) -- (-1pt,\y cm) node[anchor=east] {$\y$};
\fill[red] (4, 0) circle[radius=1.3mm];
\node[] at (3.6,0.1) {$r$};
\node[] at (4,3.5) {$a$};
\node[] at (7,3.5) {$b$};
\node[] at (4,4.5) {$e$};
\node[white] at (7.5,4.5) {$f$};
\end{tikzpicture}
        \caption{An observable cone projection}
    \end{subfigure}%
    ~~~~~~~~ 
    \begin{subfigure}[t]{0.5\textwidth}
        \centering
\begin{tikzpicture}[scale=.6]
\draw[step=1cm,gray,very thin] (0,0) grid (11,6);
\filldraw[fill=black, draw=black] (0,2) rectangle (4,3);
\filldraw[fill=black, draw=black] (7,3) rectangle (9,0);
\filldraw[fill=black, draw=black] (7,5) rectangle (11,4);
\draw[red, very thick] (4,3) -- (7,3);
\draw[red, very thick, dashed] (4,0) .. controls (4.5,2) and (4.1,3) .. (4, 3);
\draw[red, very thick, dashed] (4,0) .. controls (5,.5) and (7,2.6) .. (7,3);
\draw[blue, very thick, dashed] (7,3) .. controls (7.1,3.3) and (7.1,3.3) .. (7.4,4);
\draw[blue, very thick, dashed] (7,3) -- (11,3.9);
\draw[line width=1.5mm, blue] (7.4,4) -- (11,4);
\fill[white] (7, 3) circle[radius=1.9mm];
\fill[blue] (7, 3) circle[radius=1.5mm];
\foreach \x in {0,1,2,3,4,5,6,7,8,9,10,11}
   \draw (\x cm,1pt) -- (\x cm,-1pt) node[anchor=north] {$\x$};
\foreach \y in {0,1,2,3,4,5,6}
    \draw (1pt,\y cm) -- (-1pt,\y cm) node[anchor=east] {$\y$};
\fill[red] (4, 0) circle[radius=1.3mm];
\node[] at (3.6,0.1) {$r$};
\node[] at (4,3.5) {$a$};
\node[] at (7,3.5) {$b$};
\node[] at (11.3,4.5) {$g$};
\node[white] at (7.5,4.5) {$f$};
\end{tikzpicture}
        \caption{A non-observable cone projection}
    \end{subfigure}
    \begin{subfigure}[t]{0.5\textwidth}
        \centering
\begin{tikzpicture}[scale=.6]
\draw[step=1cm,gray,very thin] (0,0) grid (11,6);
\filldraw[fill=black, draw=black] (0,2) rectangle (4,3);
\filldraw[fill=black, draw=black] (7,3) rectangle (9,0);
\filldraw[fill=black, draw=black] (7,5) rectangle (11,4);
\draw[red, very thick] (4,3) -- (7,3);
\draw[red, very thick, dashed] (4,0) .. controls (4.5,2) and (4.1,3) .. (4,3);
\draw[red, very thick, dashed] (4,0) .. controls (5,.5) and (7,2.6) .. (7,3);
\draw[line width=1.5mm, blue] (7,3) -- (9,3);
\fill[white] (7, 3) circle[radius=1.9mm];
\fill[blue] (7, 3) circle[radius=1.5mm];
\foreach \x in {0,1,2,3,4,5,6,7,8,9,10,11}
   \draw (\x cm,1pt) -- (\x cm,-1pt) node[anchor=north] {$\x$};
\foreach \y in {0,1,2,3,4,5,6}
    \draw (1pt,\y cm) -- (-1pt,\y cm) node[anchor=east] {$\y$};
\fill[red] (4, 0) circle[radius=1.3mm];
\node[] at (3.6,0.1) {$r$};
\node[] at (4,3.5) {$a$};
\node[] at (7,3.5) {$b$};
\node[] at (9,3.5) {$c$};
\end{tikzpicture}
        \caption{An intermediate non-observable cone projection}
    \end{subfigure}%
    ~~~~~~~~ 
    \begin{subfigure}[t]{0.5\textwidth}
        \centering
\begin{tikzpicture}[scale=.6]
\draw[step=1cm,gray,very thin] (0,0) grid (11,6);
\filldraw[fill=black, draw=black] (0,2) rectangle (4,3);
\filldraw[fill=black, draw=black] (7,3) rectangle (9,0);
\filldraw[fill=black, draw=black] (7,5) rectangle (11,4);
\draw[red, very thick] (4,3) -- (7,3);
\draw[blue, very thick, dashed] (7,3) -- (9,3);
\draw[red, very thick, dashed] (4,0) .. controls (4.5,2) and (4.1,3) .. (4,3);
\draw[red, very thick, dashed] (4,0) .. controls (5,.5) and (7,2.6) .. (7,3);
\draw[line width=1.5mm, blue] (9,3) -- (11,3);
\fill[white] (7, 3) circle[radius=1.9mm];
\fill[blue] (7, 3) circle[radius=1.5mm];
\foreach \x in {0,1,2,3,4,5,6,7,8,9,10,11}
   \draw (\x cm,1pt) -- (\x cm,-1pt) node[anchor=north] {$\x$};
\foreach \y in {0,1,2,3,4,5,6}
    \draw (1pt,\y cm) -- (-1pt,\y cm) node[anchor=east] {$\y$};
\fill[red] (4, 0) circle[radius=1.3mm];
\node[] at (3.6,0.1) {$r$};
\node[] at (4,3.5) {$a$};
\node[] at (7,3.5) {$b$};
\node[] at (9,3.5) {$c$};
\node[] at (11.3,3.5) {$d$};
\end{tikzpicture}
\caption{A non-intermediate non-observable cone projection}
\end{subfigure}
    \begin{subfigure}[t]{0.5\textwidth}
        \centering
\begin{tikzpicture}[scale=1.8]
\draw[step=1cm,gray,very thin] (0,0) grid (5,3);
\filldraw[fill=black, draw=black] (0,2) rectangle (2,3);
\filldraw[fill=black, draw=black] (4,3) rectangle (5,0);
\draw[red, very thick] (2,2) -- (4,2);
\draw[red, very thick, dashed] (2,0) .. controls (1.75,1) and (1.75,1) .. (2,2);
\draw[red, very thick, dashed] (2,0) .. controls (2.2,1) and (3.5,1.8) .. (4.,2);
\draw[blue, very thick, dashed] (2,2) .. controls (1.85,2.5) and (1.85,2.5) .. (2,3);
\draw[blue, very thick, dashed] (2,2) .. controls (2,2.2) and (2,2.7) .. (2.43,3);
\draw[line width=1.3mm, blue] (2,3) -- (2.4,3);
\fill[white] (2, 2) circle[radius=1.1mm];
\fill[white] (2, 0) circle[radius=1.1mm];
\fill[blue] (2, 2) circle[radius=.8mm];
\foreach \x in {0,1,2,3,4,5}
   \draw (\x cm,1pt) -- (\x cm,-1pt) node[anchor=north] {$\x$};
\foreach \y in {0,1,2,3}
    \draw (1pt,\y cm) -- (-1pt,\y cm) node[anchor=east] {$\y$};
\fill[red] (2, 0) circle[radius=.8mm];
\node[] at (1.8,0.1) {$r$};
\node[white] at (1.75,2.2) {$a$};
\node[] at (3.9,2.2) {$b$};
\node[] at (2,3.17) {$c$};
\node[] at (2.4,3.2) {$d$};
\end{tikzpicture}
\caption{A non-observable adjoint cone projection. Here the adjoint successor of the search node $[a, b], r$ is the adjoint node $[c, d), a$, whereby all points in [c, d) are invisible from $a$ and the interval endpoint $d$ is determined such that $d$ is visible from $a$.}
\end{subfigure}
    \caption{Spherical Anya projections on a world map in Spherical geometry.}
    \label{fig:spherical_anya_projections}
\end{figure*}
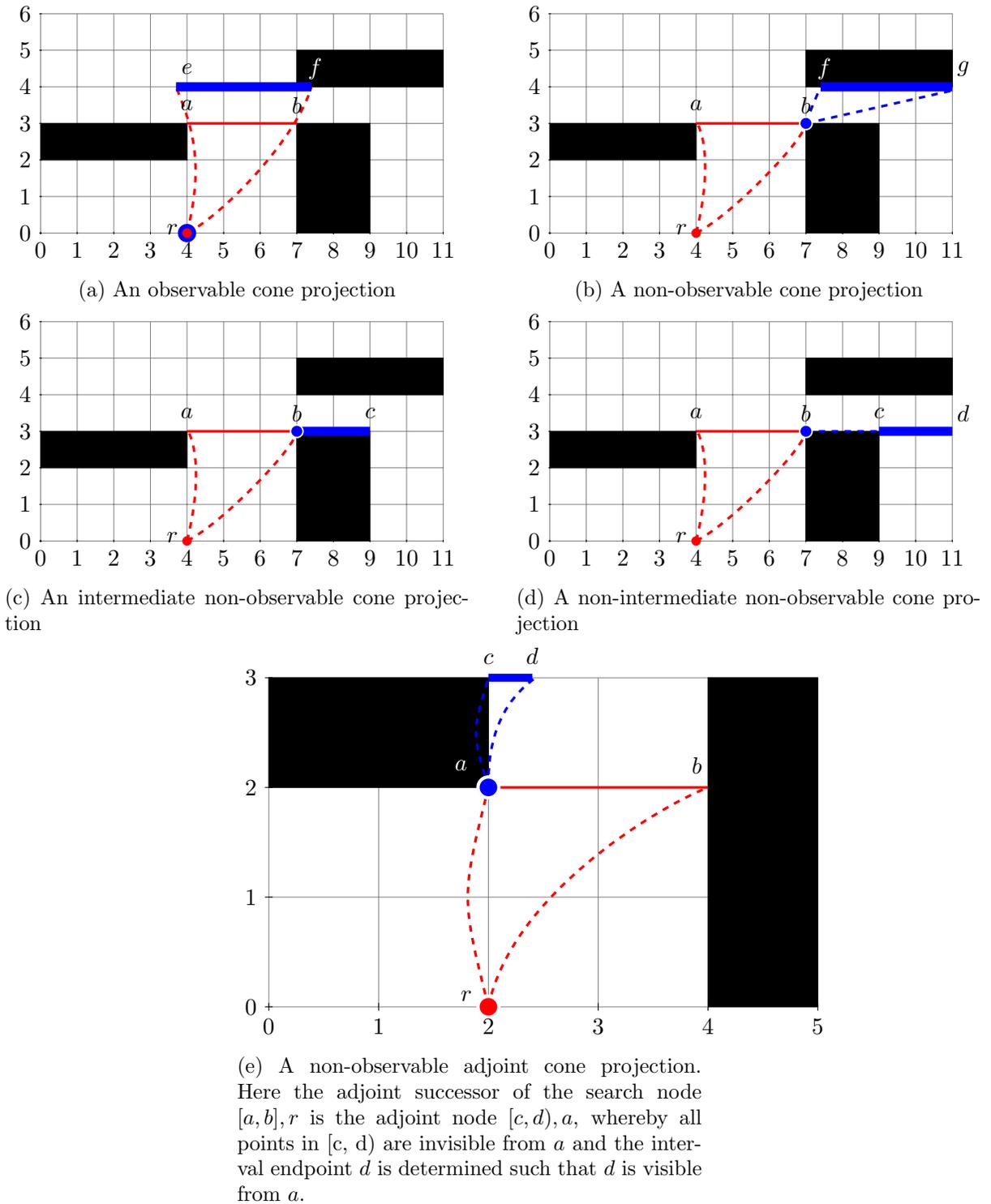

An \textit{observable cone projection} is illustrated in Figure {\ref{fig:anya_projections}}(a). This is the line-of-sight straight line projection in Euclidean space which extends the visibility cone described by the search node $[a, b], r$ from an initial row to the row above, resulting in successor node $[e, f], r$. Figure {\ref{fig:anya_projections}}(b) provides an example of a \textit{non-observable cone projection} which results in successor $(f, g], b$. Here the Euclidean straight line segment $L(b, f)$ defines an extremal line of the visibility cone for the successor of Figure {\ref{fig:anya_projections}}(a) and the interval endpoint $g$ is defined such that the length of the interval $(f, g]$ is maximal, as per the definition of a successor node. The final manifestations of cone projections are those as illustrated in Figures {\ref{fig:anya_projections}} (c) and (d), the former being an \textit{intermediate non-observable cone projection} and the latter a (non-intermediate) \textit{non-observable cone projection}.

\textit{Flat projections}, those extended from flat nodes whereby the root point lies on the same row as the interval, are typically sub-divided into two types, being \textit{intermediate} flat projections and \textit{non-intermediate} flat projections, these being referred to simply as \textit{flat projections}; see Figures {\ref{fig:anya_projections}} (e) and (f) for illustrations. An intermediate flat projection gives rise to the successor search node $(b, c], a$ as in Figure {\ref{fig:anya_projections}}(c). Due to this interval being adjacent to a non-traversable obstacle, it is known in advance that this search node will either be completely traversed or ignored, as any turning point along an optimal route cannot lie within the interval $(b, c)$. Hence one arrives at Figure {\ref{fig:anya_projections}}(d), illustrating the converse which is a non-intermediate flat projection. Here, the successor node $(c, d], a$ may not merely be traversed but instead may be the site of a turning point, leading to successive cone projections. Note that an interval may be a member of multiple successor nodes, as is the case for intervals $(b, c]$ and $(c, d]$.

\subsection{Projections in Spherical Anya}
Projections in Spherical Anya are analogous to those in Anya but for the fact that line-of-sight checking considers the great-circle line passing through point pairs. This differentiation is further exposed by the appearance of adjoint successors in Spherical Anya.

Flat projections remain the same, visually, due to the convention taking rows as lines of longitude. Where the extremal great-circle line segments $C(r, a)$ and $C(r, b)$ of a node $[a, b], r$ do not intersect an obstacle, a great-circle observable cone projection yields the standard successor $[e, f], r$ much as per Anya, save for the visibility cone defining the maximal interval $I=[e, f]$ now being calculated as the intercept of a great circle line through $r$ and $a$ with the corresponding row of the grid, as illustrated in Figure \ref{fig:spherical_anya_projections}(a). Similarly for the (great-circle) non-observable cone projection $(f, g], b$ of \ref{fig:spherical_anya_projections}(b), defining a standard successor node whereby $f$ is the last point in the cone projection $[e, f], r$ visible to $r$ and the interval $(f, g]$ is calculated to be maximal as per the definition of a successor node. Figures \ref{fig:spherical_anya_projections}(c) and (d) illustrate non-observable cone projections, these being visually identical to Anya. Only the underlying metric for the geodesic between points on the same row has changed.
Finally \textit{adjoint cone projections}, which are non-observable cone projections unique to Spherical Anya, arise when a successor is an adjoint successor, as illustrated in Figure \ref{fig:spherical_anya_projections}(e).

\section{Spherical Anya Algorithm}
\label{spherical_anya}
The recipe for Spherical Anya can now be presented using the foundational components and principles above. Globally, Spherical Anya is defined as per Anya, the difference being in the successors generated and the underlying geodesic and line-of-sight considerations. For easy reference, the Anya algorithm as per \cite{harabor2} is reproduced below.

\noindent\rule{\textwidth}{2pt}
Algorithm 1: Anya and Spherical Anya\newline
\noindent\rule{\textwidth}{1pt}
\begin{algorithmic}[1]
\Require Grid, source location $s$, target location $t$, initial root point $r_0$.
\State $open \gets \{(I = [s], r_0)\}$\algorithmiccomment{root point $r_0$ located off the grid}

\While{$open$ is not empty}
\State $(I, r) \gets pop$(open)
\If{$t \in I$}

\State \Return path\_to$(I)$
\EndIf

\ForAll{$(I', r') \in$ successors $(I, r)$}

\If {$\neg$should\_prune($I', r'$)}
        \State $open \gets open \cup\{(I', r')\}$ \algorithmiccomment{Successor pruning}
\EndIf

\EndFor

\EndWhile

\State \Return $null$

\end{algorithmic}
\noindent\rule{\textwidth}{1pt}\newline\newline
\noindent The algorithm for calculating successors in Spherical Anya differs in form from that of Anya due to the existence of adjoint nodes in Spherical geometry (see Algorithms 2 and 3). In addition, the operations {\tt generate-cone-successors} and {\tt generate-flat-successors} from \cite{harabor2} must be replaced by their great-circle analogs from Spherical Anya utilising great-circle geodesic and line-of-sight considerations. Furthermore, while the pruning strategies for Anya as described in \cite{harabor2} carry across to Spherical Anya, additional pruning strategies may be applied to adjoint successors (see Algorithm 4).\newline

\noindent\rule{\textwidth}{2pt}
Algorithm 2: Successors in Spherical Anya\newline
\noindent\rule{\textwidth}{1pt}
\begin{algorithmic}[1]
\Function{SUCCESSORS}{$n=(I, r)$} \algorithmiccomment{Input is current node}
\If{$n$ is the start node $s$}
\State \Return \texttt{generate-start-successors}($I = $[$s$])
\EndIf
\State $successors \gets \emptyset$

\If{$n$ is a flat node}

\State $p \gets $ endpoint of $I$ farthest from $r$ \algorithmiccomment{Successor interval starts from $p$}

\State $successors \gets$ \texttt{generate-flat-successors}($p, r$) \algorithmiccomment{Observable successors}

\If{$p$ is a corner turning point on a locally optimal path beginning at $r$}
\State $successors \gets successors ~~\cup$ \texttt{generate-cone-successors}($p, p, r$) \algorithmiccomment{Non-observable successors}
\EndIf

\If{$p$ is an adjoint turning point on a locally optimal path beginning at $r$}
\State $successors \gets successors ~~\cup$ \texttt{generate-adjoint-successors}($p, r$) \algorithmiccomment{Non-observable successors}
\EndIf

\ElsIf{$n$ it is a cone}

\State $a \gets$ left endpoint of $I$
\State $b \gets$ right endpoint of $I$
\State $successors \gets$ \texttt{generate-cone-successors}($a, b, r$) \algorithmiccomment{Observable successors}

\If{$a$ is a corner turning point on a locally optimal path beginning at $r$}
\State $successors \gets successors ~~\cup $ \texttt{generate-flat-successors}($a, r$) \algorithmiccomment{Non-observable successors}
\State $successors \gets successors ~~\cup $ \texttt{generate-cone-successors}($a, a, r$)\algorithmiccomment{Non-observable successors}
\EndIf

\If{$a$ is an adjoint turning point on a locally optimal path beginning at $r$}
\State $successors \gets successors ~~\cup$ \texttt{generate-adjoint-successors}($a, r$) \algorithmiccomment{Non-observable successors}
\EndIf

\If{$b$ is a corner turning point on a locally optimal path beginning at $r$}
\State $successors \gets successors ~~\cup $ \texttt{generate-flat-successors}($b, r$) \algorithmiccomment{Non-observable successors}
\State $successors \gets successors ~~\cup $ \texttt{generate-cone-successors}($b, b, r$) \algorithmiccomment{Non-observable successors}
\EndIf

\If{$b$ is an adjoint turning point on a locally optimal path beginning at $r$}
\State $successors \gets successors ~~\cup$ \texttt{generate-adjoint-successors}($b, r$) \algorithmiccomment{Non-observable successors}
\EndIf

\Else \algorithmiccomment{If node is neither flat or cone, it is an adjoint}

\State $p \gets $ endpoint of $I$ closest to $r$ 
\State $successors \gets successors ~~\cup$ \texttt{generate-adjoint-successors}($p, r$) \algorithmiccomment{Non-observable successors}
\State $successors \gets$ \texttt{generate-cone-successors}($p, p, r$) \algorithmiccomment{Non-observable successors}

\EndIf

\EndFunction
\end{algorithmic}
\noindent\rule{\textwidth}{1pt}

\subsection{Adjoint Successors}\label{section:adjoint_successors}
Adjoint successors occur only adjacent to the vertical edge of an obstacle, this by convention being defined as a line of constant latitude. Adjoint nodes in the northern hemisphere exist adjacent to the south-facing edge of an obstacle and symmetrically in the southern hemisphere exist adjacent to the north-facing edge of an obstacle (recall Figure \ref{fig:great_circle_at_edges}). Note also that an adjoint node for the same obstacle edge can share a root point with an  adjoint node adjacent to the same edge, the difference being only the row on which each of the respective intervals associated with the adjoint nodes lie. The following algorithm provides the recipe with which to calculate adjoint successors.\newline

Figure \ref{fig:intermediate_adjoint} provides an illustration of the procedure, whereby all adjoint successors of the search node $[a, b], r$ are calculated. The first adjoint successor of this search node is $[c, d), a$, the interval $[c, d)$ calculated as per the definition of an adjoint node, such that the width of the interval is maximal and hence such that $d$ is visible from the new root point $a$ but all points $p \in [c, d)$ are invisible to it. This adjoint node, having only a single possible successor, namely $[e, f), a$, is an intermediate adjoint node, as is $[e, f), a$ itself. Here, the extremal point $f$ is simply the intersection of the great circle line through $a$ and $d$ with the subsequent row of the grid. Finally, projecting the great-circle lines to the next row yields the last adjoint successor $[g, h), a$ associated with this obstacle edge. Search node $[g, h), a$ is not intermediate, having more than one successor and also does not yield further adjoint node successors.

\noindent\rule{\textwidth}{2pt}
Algorithm 3: Computing a set of adjoint successors\newline
\noindent\rule{\textwidth}{1pt}
\begin{algorithmic}[1]
\Function{Generate-Adjoint-Successors}{an interval endpoint $p$, a root point $r$}

\State \textit{successors} $\gets \emptyset$

\If {$p$ is a corner point}
\State $r^\prime \gets p$
\Else
\State $r^\prime \gets r$
\EndIf

\State $p^\prime \gets$ a point from adjacent row, farthest and non-observable one from $r^\prime$
\State $I_{max} \gets $ maximum closed interval beginning at $p^\prime$ and non-observable from $r^\prime$

\ForAll{$I \in \{$split $I_{max}$ at each corner point$\}$}
\State $n^\prime \gets$ a new search node with interval $I$ and root point $r^\prime$
\State \textit{successors} $\gets$ \textit{successors} $\cup$ $I$
\EndFor

\State \Return \textit{successors}

\EndFunction
\end{algorithmic}
\noindent\rule{\textwidth}{1pt}\newline

\begin{figure*}[h!]
    \centering
\begin{tikzpicture}[scale=1.5]
\draw[step=1cm, gray, very thin] (0,0) grid (5,6);
\filldraw[fill=black, draw=black] (0,2) rectangle (2,4);
\filldraw[fill=black, draw=black] (1,4) rectangle (2,5);
\filldraw[fill=black, draw=black] (4,3) rectangle (5,0);
\draw[red, very thick] (2,2) -- (4,2);
\draw[red, very thick, dashed] (2,0) .. controls (1.85,1) and (1.75,1) .. (2,2);
\draw[red, very thick, dashed] (2,0) .. controls (2.2,1) and (3,1.8) .. (4,2);
\draw[blue, very thick, dashed] (2,2) .. controls (1.95,2.5) and (1.85,2.5) .. (2,3);
\draw[blue, very thick, dashed] (2,2) .. controls (1.85,3) and (1.75,3) .. (2,4);
\draw[blue, very thick, dashed] (2,2) .. controls (1.7,3.5) and (1.6,3.5) .. (2,5);
\draw[blue, very thick, dashed] (2,2) .. controls (2.1,3.5) and (3,4.5) .. (4.8,5);
\draw[line width=0.8mm, blue] (2,3) -- (2.22,3);
\draw[line width=0.8mm, blue] (2,4) -- (2.9,4);
\draw[line width=1.3mm, blue] (2,5) -- (4.9,5);
\fill[white] (2, 2) circle[radius=1.1mm];
\fill[white] (2, 0) circle[radius=1.1mm];
\fill[blue] (2, 2) circle[radius=.8mm];
\foreach \x in {0,1,2,3,4,5}
   \draw (\x cm,1pt) -- (\x cm,-1pt) node[anchor=north] {$\x$};
\foreach \y in {0,1,2,3,4,5,6}
    \draw (1pt,\y cm) -- (-1pt,\y cm) node[anchor=east] {$\y$};
\fill[red] (2, 0) circle[radius=.8mm];
\node[] at (1.8,0.1) {$r$};
\node[white] at (1.75,2.2) {$a$};
\node[] at (3.9,2.2) {$b$};
\node[white] at (1.9,3.15) {$c$};
\node[] at (2.2,3.17) {$d$};
\node[white] at (1.9,4.15) {$e$};
\node[] at (2.95,4.17) {$f$};
\node[] at (2,5.2) {$g$};
\node[] at (4.88,5.2) {$h$};
\end{tikzpicture}
    \caption{A series of intermediate adjoint nodes which are calculated in the process of generating the final non-intermediate adjoint node $[g, h), a$.}
    \label{fig:intermediate_adjoint}
\end{figure*}
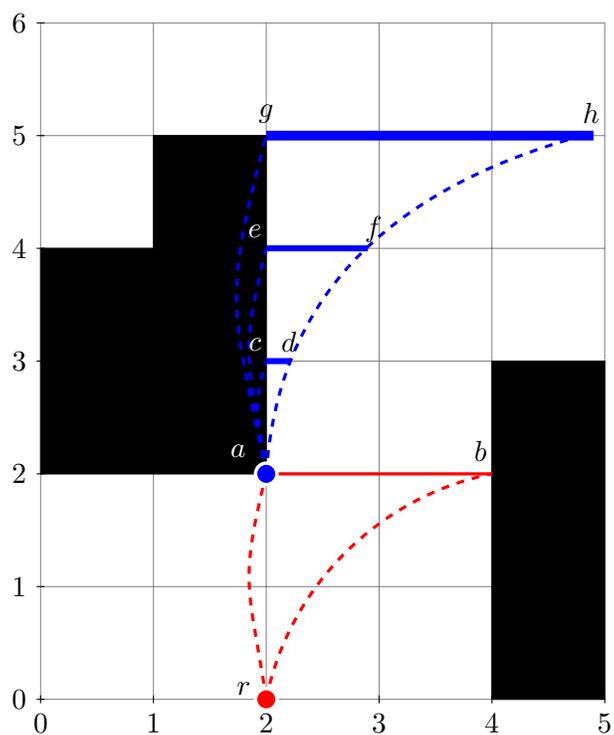

\noindent\rule{\textwidth}{2pt}
Algorithm 4: Intermediate node pruning\newline
\noindent\rule{\textwidth}{1pt}
\begin{algorithmic}[1]
\Function{Is-Intermediate}{$n=(I, r)$}
\If{$n$ is a flat node}

\State $p \gets$ endpoint of $I$ furthest from $r$

\If{$p$ is a corner turning point for a locally optimal path with prefix $\langle r, p \rangle$}
\State \Return \textit{false} \algorithmiccomment{$n$ has at least one non-observable successor; it cannot be intermediate}
\EndIf 

\ElsIf{$n$ is an adjoint node}
\State $p^\prime \gets$ endpoint of $I$ closest to $r$
\If{$p^\prime$ is a corner point}
\State \Return \textit{false} \algorithmiccomment{$n$ has at least one non-observable successor; it cannot be intermediate}
\EndIf

\Else
\If{$I$ has a closed endpoint that is also a corner point}
\State \Return \textit{false} \algorithmiccomment{$n$ has at least one non-observable successor; it cannot be intermediate}
\EndIf
\State $I^\prime \gets$ interval after projecting $r$ through $I$
\If{$I^\prime$ contains any corner points}
\State \Return \textit{false} \algorithmiccomment{$n$ has at least one non-observable successor; it cannot be intermediate}
\EndIf

\EndIf

\State \Return \textit{true}

\EndFunction
\end{algorithmic}
\noindent\rule{\textwidth}{1pt}

\subsection{Spherical Anya Algorithm: Examples}
We provide examples which illustrate how Spherical Anya works in practice by illustrating all successors for a given root node as would be necessary in an implementation of the Spherical Anya algorithm, search nodes and their corresponding successors comprising the fundamental units of the search algorithm.

\subsubsection{Successors of a flat search node}
In the following scenario the Spherical Anya algorithm is required to generate all successors of the flat search node $(r, a], r \equiv (0, 2], r$, of Figure \ref{fig:spherical_anya_example_1}. As the root point of the node is on the same row as the interval, the first projections are flat projections, generating flat successors until the wall of Figure \ref{fig:spherical_anya_example_1}(d) interrupts progress. In this instance the wall is created by the existence of two diagonally adjacent non-traversable cells centred on the point $e=(7, 2)$ of Figure \ref{fig:spherical_anya_example_1}(d). Not all of these successors survive however, those not containing the target point and being intermediate by virtue of their adjacency to a non-traversable cell, as $(b, c], r$ is in Figure \ref{fig:spherical_anya_example_1}(b) and $(d, e], r$ is in Figure \ref{fig:spherical_anya_example_1}(d), being discarded. These cannot precede a turning point on any optimal route. The search node $(d, e], r$ of Figure \ref{fig:spherical_anya_example_1}(d) is in addition a cul-de-sac node. It has no successors, and the target does not lie in the interval $(d, e]$.

All flat successors being thus generated, the Spherical Anya algorithm generates all flat-to-adjoint projections and flat-to-cone projections. These occur at the farthest corner of each non-traversable interval from the root $r$, along the same row and prior to the edge defined by the wall. In this example, flat-to-adjoint and flat-to-cone projections will be generated by Spherical Anya at each of the new root points $a = (2, 2)$ and $c = (5, 2)$, as illustrated in figures \ref{fig:spherical_anya_example_1}(e)--(h).

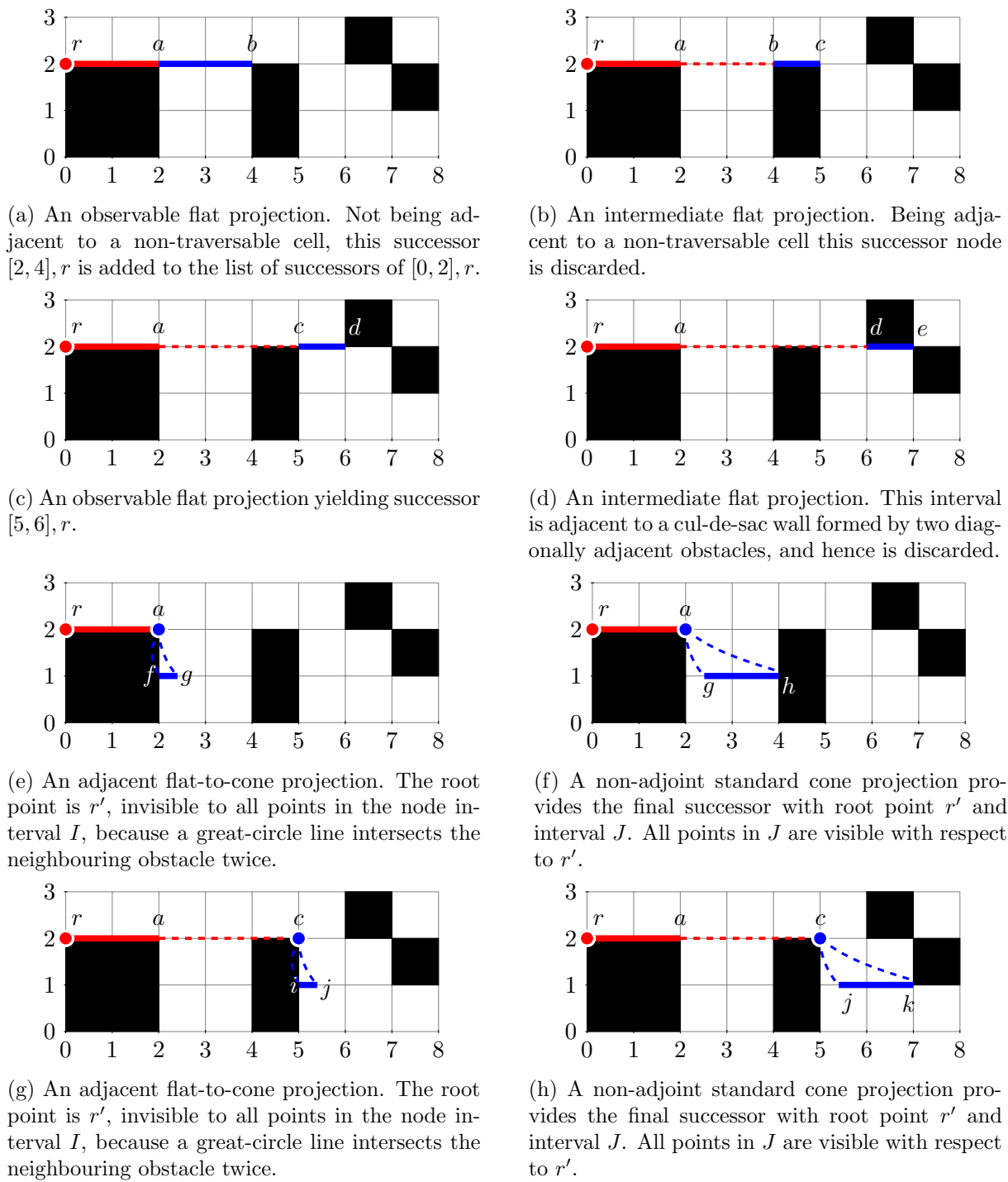
\begin{figure*}[h!]
    \centering
    \begin{subfigure}[t]{0.5\textwidth}
        \centering
\begin{tikzpicture}[scale=.75]
\draw[step=1cm,gray,very thin] (0,0) grid (8,3);
\filldraw[fill=black, draw=black] (0,0) rectangle (2,2);
\filldraw[fill=black, draw=black] (4,0) rectangle (5,2);
\filldraw[fill=black, draw=black] (6,2) rectangle (7,3);
\filldraw[fill=black, draw=black] (7,1) rectangle (8,2);
\draw[line width=1mm, red] (0,2) -- (2,2);
\draw[line width=1mm, blue] (2,2) -- (4,2);
\fill[white] (0, 2) circle[radius=1.9mm];
\foreach \x in {0,1,2,3,4,5,6,7,8}
   \draw (\x cm,1pt) -- (\x cm,-1pt) node[anchor=north] {$\x$};
\foreach \y in {0,1,2,3}
   \draw (1pt,\y cm) -- (-1pt,\y cm) node[anchor=east] {$\y$};
\fill[red] (0, 2) circle[radius=1.3mm];
\node[] at (0.25, 2.4) {$r$};
\node[] at (2, 2.4) {$a$};
\node[] at (4, 2.45) {$b$};

\end{tikzpicture}
        \caption{An observable flat projection. Not being adjacent to a non-traversable cell, this successor $[2, 4], r$ is added to the list of successors of $[0, 2], r$.}
    \end{subfigure}%
    ~~~~~~~~
    \begin{subfigure}[t]{0.5\textwidth}
        \centering
\begin{tikzpicture}[scale=.75]
\draw[step=1cm,gray,very thin] (0,0) grid (8,3);
\filldraw[fill=black, draw=black] (0,0) rectangle (2,2);
\filldraw[fill=black, draw=black] (4,0) rectangle (5,2);
\filldraw[fill=black, draw=black] (6,2) rectangle (7,3);
\filldraw[fill=black, draw=black] (7,1) rectangle (8,2);
\draw[line width=1mm, red] (0,2) -- (2,2);
\draw[line width=.5mm, red, dashed] (2,2) -- (5,2);
\draw[line width=1mm, blue] (4,2) -- (5,2);
\fill[white] (0, 2) circle[radius=1.9mm];
\foreach \x in {0,1,2,3,4,5,6,7,8}
   \draw (\x cm,1pt) -- (\x cm,-1pt) node[anchor=north] {$\x$};
\foreach \y in {0,1,2,3}
    \draw (1pt,\y cm) -- (-1pt,\y cm) node[anchor=east] {$\y$};
\fill[red] (0, 2) circle[radius=1.3mm];
\node[] at (0.25, 2.4) {$r$};
\node[] at (2, 2.4) {$a$};
\node[] at (4, 2.45) {$b$};
\node[] at (5, 2.4) {$c$};

\end{tikzpicture}
        \caption{An intermediate flat projection. Being adjacent to a non-traversable cell this successor node is discarded.}
    \end{subfigure}%

    \begin{subfigure}[t]{0.5\textwidth}
        \centering
\begin{tikzpicture}[scale=.75]
\draw[step=1cm,gray,very thin] (0,0) grid (8,3);
\filldraw[fill=black, draw=black] (0,0) rectangle (2,2);
\filldraw[fill=black, draw=black] (4,0) rectangle (5,2);
\filldraw[fill=black, draw=black] (6,2) rectangle (7,3);
\filldraw[fill=black, draw=black] (7,1) rectangle (8,2);
\draw[line width=1mm, red] (0,2) -- (2,2);
\draw[line width=.5mm, red, dashed] (2,2) -- (5,2);
\draw[line width=1mm, blue] (5,2) -- (6,2);
\fill[white] (0, 2) circle[radius=1.9mm];
\foreach \x in {0,1,2,3,4,5,6,7,8}
   \draw (\x cm,1pt) -- (\x cm,-1pt) node[anchor=north] {$\x$};
\foreach \y in {0,1,2,3}
    \draw (1pt,\y cm) -- (-1pt,\y cm) node[anchor=east] {$\y$};
\fill[red] (0, 2) circle[radius=1.3mm];
\node[] at (0.25, 2.4) {$r$};
\node[] at (2, 2.4) {$a$};
\node[] at (5, 2.4) {$c$};
\node[white] at (6.2, 2.44) {$d$};

\end{tikzpicture}
        \caption{An observable flat projection yielding successor $[5, 6], r$.}
    \end{subfigure}%
    ~~~~~~~~
    \begin{subfigure}[t]{0.5\textwidth}
        \centering
\begin{tikzpicture}[scale=.75]
\draw[step=1cm,gray,very thin] (0,0) grid (8,3);
\filldraw[fill=black, draw=black] (0,0) rectangle (2,2);
\filldraw[fill=black, draw=black] (4,0) rectangle (5,2);
\filldraw[fill=black, draw=black] (6,2) rectangle (7,3);
\filldraw[fill=black, draw=black] (7,1) rectangle (8,2);
\draw[line width=1mm, red] (0,2) -- (2,2);
\draw[line width=.5mm, red, dashed] (2,2) -- (6,2);
\draw[line width=1mm, blue] (6,2) -- (7,2);
\fill[white] (0, 2) circle[radius=1.9mm];
\foreach \x in {0,1,2,3,4,5,6,7,8}
   \draw (\x cm,1pt) -- (\x cm,-1pt) node[anchor=north] {$\x$};
\foreach \y in {0,1,2,3}
    \draw (1pt,\y cm) -- (-1pt,\y cm) node[anchor=east] {$\y$};
\fill[red] (0, 2) circle[radius=1.3mm];
\node[] at (0.25, 2.4) {$r$};
\node[] at (2, 2.4) {$a$};
\node[white] at (6.2, 2.44) {$d$};
\node[] at (7.2, 2.34) {$e$};
\end{tikzpicture}
        \caption{An intermediate flat projection. This interval is adjacent to a cul-de-sac wall formed by two diagonally adjacent obstacles, and hence is discarded.}
    \end{subfigure}%

        \begin{subfigure}[t]{0.5\textwidth}
        \centering
\begin{tikzpicture}[scale=.75]
\draw[step=1cm,gray,very thin] (0,0) grid (8,3);
\filldraw[fill=black, draw=black] (0,0) rectangle (2,2);
\filldraw[fill=black, draw=black] (4,0) rectangle (5,2);
\filldraw[fill=black, draw=black] (6,2) rectangle (7,3);
\filldraw[fill=black, draw=black] (7,1) rectangle (8,2);
\draw[line width=1mm, red] (0,2) -- (2,2);
\draw[line width=1mm, blue] (2,1) -- (2.4,1);

\draw[blue, very thick, dashed] (2,2) .. controls (1.8,1.5) and (1.8,1.5) .. (2,1);
\draw[blue, very thick, dashed] (2,2) .. controls (2,1.4) and (2.4,1) .. (2.4,1);
\fill[white] (0, 2) circle[radius=1.9mm];
\fill[white] (2, 2) circle[radius=1.9mm];
\fill[blue] (2, 2) circle[radius=1.3mm];
\foreach \x in {0,1,2,3,4,5,6,7,8}
   \draw (\x cm,1pt) -- (\x cm,-1pt) node[anchor=north] {$\x$};
\foreach \y in {0,1,2,3}
   \draw (1pt,\y cm) -- (-1pt,\y cm) node[anchor=east] {$\y$};
\fill[red] (0, 2) circle[radius=1.3mm];
\node[] at (0.25, 2.4) {$r$};
\node[] at (2, 2.4) {$a$};
\node[white] at (1.8, 1) {$f$};
\node[] at (2.6, .95) {$g$};

\end{tikzpicture}
        \caption{An adjacent flat-to-cone projection. The root point is $r'$, invisible to all points in the node interval $I$, because a great-circle line intersects the neighbouring obstacle twice.}
    \end{subfigure}%
    ~~~~~~~~~
    \begin{subfigure}[t]{0.5\textwidth}
        \centering
\begin{tikzpicture}[scale=.75]
\draw[step=1cm,gray,very thin] (0,0) grid (8,3);
\filldraw[fill=black, draw=black] (0,0) rectangle (2,2);
\filldraw[fill=black, draw=black] (4,0) rectangle (5,2);
\filldraw[fill=black, draw=black] (6,2) rectangle (7,3);
\filldraw[fill=black, draw=black] (7,1) rectangle (8,2);
\draw[line width=1mm, red] (0,2) -- (2,2);
\draw[line width=1mm, blue] (2.4,1) -- (4,1);
\draw[blue, very thick, dashed] (2,2) .. controls (2,1.4) and (2.4,1) .. (2.4,1);
\draw[blue, very thick, dashed] (2,2) .. controls (2.5,1.6) and (3,1.4) .. (4,1.1);

\fill[white] (0, 2) circle[radius=1.9mm];
\fill[white] (2, 2) circle[radius=1.9mm];
\fill[blue] (2, 2) circle[radius=1.3mm];
\foreach \x in {0,1,2,3,4,5,6,7,8}
   \draw (\x cm,1pt) -- (\x cm,-1pt) node[anchor=north] {$\x$};
\foreach \y in {0,1,2,3}
   \draw (1pt,\y cm) -- (-1pt,\y cm) node[anchor=east] {$\y$};
\fill[red] (0, 2) circle[radius=1.3mm];
\node[] at (0.25, 2.4) {$r$};
\node[] at (2, 2.4) {$a$};
\node[] at (2.5, .7) {$g$};
\node[white] at (4.21, .8) {$h$};
\end{tikzpicture}
        \caption{A non-adjoint standard cone projection provides the final successor with root point $r'$ and interval $J$. All points in $J$ are visible with respect to $r'$.}
    \end{subfigure}%

\begin{subfigure}[t]{0.5\textwidth}
        \centering
\begin{tikzpicture}[scale=.75]
\draw[step=1cm,gray,very thin] (0,0) grid (8,3);
\filldraw[fill=black, draw=black] (0,0) rectangle (2,2);
\filldraw[fill=black, draw=black] (4,0) rectangle (5,2);
\filldraw[fill=black, draw=black] (6,2) rectangle (7,3);
\filldraw[fill=black, draw=black] (7,1) rectangle (8,2);
\draw[line width=1mm, red] (0,2) -- (2,2);
\draw[line width=.5mm, red, dashed] (2,2) -- (5,2);
\draw[line width=1mm, blue] (5,1) -- (5.4,1);

\draw[blue, very thick, dashed] (5,2) .. controls (4.8,1.5) and (4.8,1.5) .. (5,1);
\draw[blue, very thick, dashed] (5,2) .. controls (5,1.4) and (5.4,1) .. (5.4,1);
\fill[white] (0, 2) circle[radius=1.9mm];
\fill[white] (5, 2) circle[radius=1.9mm];
\fill[blue] (5, 2) circle[radius=1.3mm];
\foreach \x in {0,1,2,3,4,5,6,7,8}
   \draw (\x cm,1pt) -- (\x cm,-1pt) node[anchor=north] {$\x$};
\foreach \y in {0,1,2,3}
   \draw (1pt,\y cm) -- (-1pt,\y cm) node[anchor=east] {$\y$};
\fill[red] (0, 2) circle[radius=1.3mm];
\node[] at (0.25, 2.4) {$r$};
\node[] at (2, 2.4) {$a$};
\node[] at (5, 2.4) {$c$};
\node[white] at (4.9, 1) {$i$};
\node[] at (5.6, .9) {$j$};

\end{tikzpicture}
        \caption{An adjacent flat-to-cone projection. The root point is $r'$, invisible to all points in the node interval $I$, because a great-circle line intersects the neighbouring obstacle twice.}
    \end{subfigure}%
    ~~~~~~~~
    \begin{subfigure}[t]{0.5\textwidth}
        \centering
\begin{tikzpicture}[scale=.75]
\draw[step=1cm,gray,very thin] (0,0) grid (8,3);
\filldraw[fill=black, draw=black] (0,0) rectangle (2,2);
\filldraw[fill=black, draw=black] (4,0) rectangle (5,2);
\filldraw[fill=black, draw=black] (6,2) rectangle (7,3);
\filldraw[fill=black, draw=black] (7,1) rectangle (8,2);
\draw[line width=1mm, red] (0,2) -- (2,2);
\draw[line width=.5mm, red, dashed] (2,2) -- (5,2);
\draw[line width=1mm, blue] (5.4,1) -- (7,1);\fill[white] (0, 2) circle[radius=1.9mm];

\draw[blue, very thick, dashed] (5,2) .. controls (5,1.4) and (5.4,1) .. (5.4,1);
\draw[blue, very thick, dashed] (5,2) .. controls (5.5,1.6) and (6,1.4) .. (7,1.1);

\fill[white] (5, 2) circle[radius=1.9mm];
\fill[blue] (5, 2) circle[radius=1.3mm];
\foreach \x in {0,1,2,3,4,5,6,7,8}
   \draw (\x cm,1pt) -- (\x cm,-1pt) node[anchor=north] {$\x$};
\foreach \y in {0,1,2,3}
   \draw (1pt,\y cm) -- (-1pt,\y cm) node[anchor=east] {$\y$};
\fill[red] (0, 2) circle[radius=1.3mm];
\node[] at (0.25, 2.4) {$r$};
\node[] at (2, 2.4) {$a$};
\node[] at (5, 2.4) {$c$};
\node[] at (5.6, .6) {$j$};
\node[] at (6.9, .6) {$k$};
\end{tikzpicture}
        \caption{A non-adjoint standard cone projection provides the final successor with root point $r'$ and interval $J$. All points in $J$ are visible with respect to $r'$.}
    \end{subfigure}%
    \begin{subfigure}[t]{0.5\textwidth}
        \centering
\begin{tikzpicture}[scale=.75]
\end{tikzpicture}
    \end{subfigure}%
\caption{All successors of search node $(r, a], r$ are illustrated to provide an example of successor generation in Spherical Anya.}
\label{fig:spherical_anya_example_1}
\end{figure*}

\subsubsection{Successors of a cone search node}
In the next example, the search node $[a, b], r$ is now a cone search node, as illustrated in Figure \ref{fig:spherical_anya_example_2}, as the root point $r$ no longer lies on the same row as the interval $[a, b]$. The first successor of this is calculated as the great-circle cone projection to the next row of the map, yielding successor $[c, d], r$ of Figure \ref{fig:spherical_anya_example_2}(a) with the same root point as its parent. Subsequent to the (non-adjoint) cone projection yielding search node $[c, d], r$ all invisible cone projections must be generated, these being invisible from root point $r$. This is achieved by forming the search node $(d, e], b$ of Figure \ref{fig:spherical_anya_example_2}(b) by searching right from $d$, the last point on row 3 visible to $r$. 

Once there are no further options for expanding a great-circle visibility cone to the third row of the map the search continues horizontally, generating flat, intermediate, flat-to-adjoint and flat-to-cone projections. This process continues along row 2 until a wall is reached, or as is the case in this scenario, the edge of the map. The first flat projection is $(b, f], b$, this being the case because the new root point $b$ is selected due to being the extremal edge of the parent interval. Node $(b, f], b$ however, being adjacent to an obstacle, is intermediate and not added to the list of successors. The subsequent successor $(f,g], b$, while not being adjacent to a non-traversable cell in this instance, is also not added to the list of successors for the parent node $[a, b], r$ because $(f,g], b$ is adjacent to the edge of the map and therefore is a successor without successors. Given the target doesn't lie in the interval $(f,g]$ it is a cul-de-sac, and hence dropped as well. The only remaining successors in this scenario are therefore the flat-to-adjoint and flat-to-cone successors of Figures \ref{fig:spherical_anya_example_2}(e) and (f).

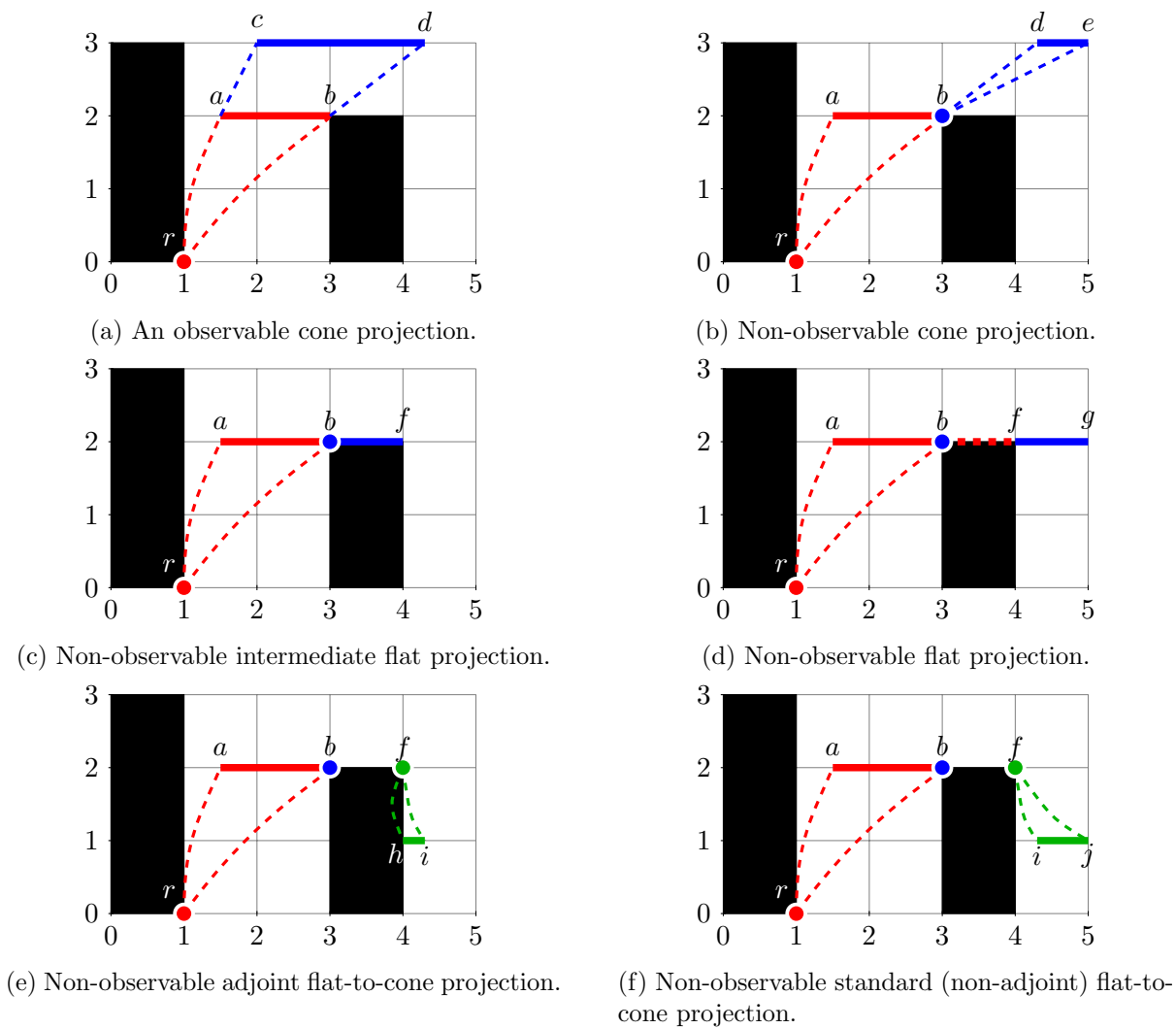
\begin{figure*}[h!]
\centering
\begin{subfigure}[t]{0.5\textwidth}
	\centering
\begin{tikzpicture}[scale=1]
\draw[step=1cm,gray,very thin] (0,0) grid (5,3);
\filldraw[fill=black, draw=black] (0,0) rectangle (1,3);
\filldraw[fill=black, draw=black] (3,0) rectangle (4,2);
\draw[line width=1mm, red] (1.5,2) -- (3,2);
\draw[line width=1mm, blue] (2,3) -- (4.3,3);
\draw[red, very thick, dashed] (1,0) .. controls (1,1) and (1.2,1.3) .. (1.5,2);
\draw[red, very thick, dashed] (1,0) .. controls (2,1.4) and (3,1.9) .. (3, 2);
\draw[blue, very thick, dashed] (1.5,2) .. controls (1.5,2) and (2,3) .. (2,3);
\draw[blue, very thick, dashed] (3,2) .. controls (3,2) and (4.3,3) .. (4.3,3);
\foreach \x in {0,1,2,3,4,5}
   \draw (\x cm,1pt) -- (\x cm,-1pt) node[anchor=north] {$\x$};
\foreach \y in {0,1,2,3}
    \draw (1pt,\y cm) -- (-1pt,\y cm) node[anchor=east] {$\y$};
\fill[white] (1, 0) circle[radius=1.5mm];
\fill[red] (1,0) circle[radius=1mm];
\node[white] at (0.8, .3) {$r$};
\node[] at (3, 2.3) {$b$};
\node[black] at (1.45, 2.25) {$a$};
\node[] at (2, 3.3) {$c$};
\node[] at (4.3, 3.3) {$d$};
\end{tikzpicture}
\caption{An observable cone projection.}
\end{subfigure}%
~~~~~~~~
\begin{subfigure}[t]{0.5\textwidth}
	\centering
\begin{tikzpicture}[scale=1]
\draw[step=1cm,gray,very thin] (0,0) grid (5,3);
\filldraw[fill=black, draw=black] (0,0) rectangle (1,3);
\filldraw[fill=black, draw=black] (3,0) rectangle (4,2);
\draw[line width=1mm, red] (1.5,2) -- (3,2);
\draw[red, very thick, dashed] (1,0) .. controls (1,1) and (1.2,1.3) .. (1.5,2);
\draw[red, very thick, dashed] (1,0) .. controls (2,1.4) and (3,1.9) .. (3, 2);
\draw[blue, very thick, dashed] (3,2) .. controls (3,2) and (4.3,3) .. (4.3,3);
\draw[blue, very thick, dashed] (3,2) .. controls (3,2) and (5,3) .. (5,3);
\draw[line width=1mm, blue] (4.3,3) -- (5,3);
\fill[white] (1, 0) circle[radius=1.5mm];
\fill[white] (3,2) circle[radius=1.5mm];
\fill[blue] (3,2) circle[radius=1mm];
\foreach \x in {0,1,2,3,4,5}
   \draw (\x cm,1pt) -- (\x cm,-1pt) node[anchor=north] {$\x$};
\foreach \y in {0,1,2,3}
    \draw (1pt,\y cm) -- (-1pt,\y cm) node[anchor=east] {$\y$};
\fill[red] (1,0) circle[radius=1mm];
\node[white] at (0.8, .3) {$r$};
\node[] at (3, 2.3) {$b$};
\node[black] at (1.5, 2.25) {$a$};
\node[] at (5, 3.25) {$e$};
\node[] at (4.3, 3.3) {$d$};
\end{tikzpicture}
\caption{Non-observable cone projection.}
\end{subfigure}%

\begin{subfigure}[t]{0.5\textwidth}
\centering
\begin{tikzpicture}[scale=1]
\draw[step=1cm,gray,very thin] (0,0) grid (5,3);
\filldraw[fill=black, draw=black] (0,0) rectangle (1,3);
\filldraw[fill=black, draw=black] (3,0) rectangle (4,2);
\draw[line width=1mm, blue] (3,2) -- (4,2);
\draw[line width=1mm, red] (1.5,2) -- (3,2);
\draw[red, very thick, dashed] (1,0) .. controls (1,1) and (1.2,1.3) .. (1.5,2);
\draw[red, very thick, dashed] (1,0) .. controls (2,1.4) and (3,1.9) .. (3, 2);
\foreach \x in {0,1,2,3,4,5}
   \draw (\x cm,1pt) -- (\x cm,-1pt) node[anchor=north] {$\x$};
\foreach \y in {0,1,2,3}
    \draw (1pt,\y cm) -- (-1pt,\y cm) node[anchor=east] {$\y$};
\node[black] at (1.5, 2.25) {$a$};
\node[] at (3, 2.3) {$b$};
\node[black] at (4, 2.3) {$f$};
\node[white] at (0.8, .3) {$r$};
\fill[white] (1, 0) circle[radius=1.5mm];
\fill[red] (1,0) circle[radius=1mm];
\fill[white] (3,2) circle[radius=1.5mm];
\fill[blue] (3,2) circle[radius=1mm];
\end{tikzpicture}
\caption{Non-observable intermediate flat projection.}
\end{subfigure}%
~~~~~~~~
\begin{subfigure}[t]{0.5\textwidth}
 \centering
\begin{tikzpicture}[scale=1]
\draw[step=1cm,gray,very thin] (0,0) grid (5,3);
\filldraw[fill=black, draw=black] (0,0) rectangle (1,3);
\filldraw[fill=black, draw=black] (3,0) rectangle (4,2);
\draw[line width=1mm, red, dashed] (3,2) -- (4,2);
\draw[line width=1mm, blue] (4,2) -- (5,2);
\draw[line width=1mm, red] (1.5,2) -- (3,2);
\draw[red, very thick, dashed] (1,0) .. controls (1,1) and (1.2,1.3) .. (1.5,2);
\draw[red, very thick, dashed] (1,0) .. controls (2,1.4) and (3,1.9) .. (3, 2);
\fill[white] (3,2) circle[radius=1.5mm];
\fill[blue] (3,2) circle[radius=1mm];
\foreach \x in {0,1,2,3,4,5}
   \draw (\x cm,1pt) -- (\x cm,-1pt) node[anchor=north] {$\x$};
\foreach \y in {0,1,2,3}
    \draw (1pt,\y cm) -- (-1pt,\y cm) node[anchor=east] {$\y$};
    \node[black] at (1.5, 2.25) {$a$};
\node[] at (3, 2.3) {$b$};
\node[] at (5, 2.3) {$g$};
\node[black] at (4, 2.3) {$f$};
\node[white] at (0.8, .3) {$r$};
\fill[white] (1, 0) circle[radius=1.5mm];
\fill[red] (1,0) circle[radius=1mm];
\end{tikzpicture}
\caption{Non-observable flat projection.}
\end{subfigure}%

\begin{subfigure}[t]{0.5\textwidth}
\centering
\begin{tikzpicture}[scale=1]
\draw[step=1cm,gray,very thin] (0,0) grid (5,3);
\filldraw[fill=black, draw=black] (0,0) rectangle (1,3);
\filldraw[fill=black, draw=black] (3,0) rectangle (4,2);
\draw[line width=1mm, black!30!green] (4, 1) -- (4.3,1);
\draw[line width=1mm, red] (1.5,2) -- (3,2);
\draw[black!30!green, very thick, dashed] (4,2) .. controls (3.8,1.5) and (3.8,1.5) .. (4,1);
\draw[black!30!green, very thick, dashed] (4,2) .. controls (4.1,1.3) and (4.1,1.3) .. (4.3,1);
\draw[red, very thick, dashed] (1,0) .. controls (1,1) and (1.2,1.3) .. (1.5,2);
\draw[red, very thick, dashed] (1,0) .. controls (2,1.4) and (3,1.9) .. (3, 2);
\fill[white] (3,2) circle[radius=1.5mm];
\fill[blue] (3,2) circle[radius=1mm];
\foreach \x in {0,1,2,3,4,5}
   \draw (\x cm,1pt) -- (\x cm,-1pt) node[anchor=north] {$\x$};
\foreach \y in {0,1,2,3}
    \draw (1pt,\y cm) -- (-1pt,\y cm) node[anchor=east] {$\y$};
\node[black] at (1.5, 2.25) {$a$};
\node[] at (4.3, .8) {$i$};
\node[] at (3, 2.3) {$b$};
\node[white] at (3.9, .83) {$h$};
\node[white] at (0.8, .3) {$r$};
\fill[white] (1, 0) circle[radius=1.5mm];
\fill[red] (1,0) circle[radius=1mm];
\fill[white] (4,2) circle[radius=1.5mm];
\fill[fill=black!30!green] (4,2) circle[radius=1mm];
\node[black] at (4, 2.25) {$f$};
\end{tikzpicture}
\caption{Non-observable adjoint flat-to-cone projection.}
\end{subfigure}%
~~~~~~~~
\begin{subfigure}[t]{0.5\textwidth}
 \centering
\begin{tikzpicture}[scale=1]
\draw[step=1cm,gray,very thin] (0,0) grid (5,3);
\filldraw[fill=black, draw=black] (0,0) rectangle (1,3);
\filldraw[fill=black, draw=black] (3,0) rectangle (4,2);
\draw[line width=1mm, red] (1.5,2) -- (3,2);
\draw[line width=1mm, black!30!green] (4.3, 1) -- (5,1);
\draw[black!30!green, very thick, dashed] (4,2) .. controls (4.1,1.3) and (4.1,1.3) .. (4.3,1);
\draw[black!30!green, very thick, dashed] (4,2) .. controls (4.5,1.3) and (4.5,1.3) .. (5,1);
\draw[red, very thick, dashed] (1,0) .. controls (1,1) and (1.2,1.3) .. (1.5,2);
\draw[red, very thick, dashed] (1,0) .. controls (2,1.4) and (3,1.9) .. (3, 2);
\fill[white] (3,2) circle[radius=1.5mm];
\fill[blue] (3,2) circle[radius=1mm];
\fill[white] (4,2) circle[radius=1.5mm];
\fill[fill=black!30!green] (4,2) circle[radius=1mm];
\foreach \x in {0,1,2,3,4,5}
   \draw (\x cm,1pt) -- (\x cm,-1pt) node[anchor=north] {$\x$};
\foreach \y in {0,1,2,3}
    \draw (1pt,\y cm) -- (-1pt,\y cm) node[anchor=east] {$\y$};
\node[black] at (1.5, 2.25) {$a$};
\node[] at (4.3, .8) {$i$};
\node[] at (5, .8) {$j$};
\node[black] at (4, 2.25) {$f$};
\node[] at (3, 2.3) {$b$};
\node[white] at (0.8, .3) {$r$};
\fill[white] (1, 0) circle[radius=1.5mm];
\fill[red] (1,0) circle[radius=1mm];
\end{tikzpicture}
\caption{Non-observable standard (non-adjoint) flat-to-cone projection.}
\end{subfigure}%
\caption{A series of flat, flat-to-adjoint and flat-to-cone projections define the successors to the illustrated search node $[a, b], r$.}
\label{fig:spherical_anya_example_2}
\end{figure*}

\section{Results}
\label{Results}
Results are reported on execution time (ET) and number of tiles crossed (CT), a tile being a single unit of the world map grid. Statistics are calculated for 4 game maps plus city/street maps and four random maps, in addition to the Earth maps determined by the NOAA bathymetric dataset at several resolutions. Comparisons and benchmarks to existing Euclidean approaches being already reported in \cite{harabor2}, the following results compare exclusively the difference between Anya and Spherical Anya. Both algorithms are implemented in Python and executed on an AWS EC2 Dual Core Intel Xeon Platinum 8000 ``R5.Large'' Instance with 16GB of RAM. The definition of the benchmark is designed to measure the relative expense in time and route length of achieving optimality in Spherical geometry with the application of Spherical Anya, compared to applying Anya with the assumption that the underlying geometry is Euclidean. We highlight that such a comparison is dependent on programming language and employed data structures, making all benchmark statistics herein indicative in Python only. We do not provide a comparison across programming languages, for which these figures would likely be different. For any world map defined as a list of coordinate pairs $[(x_1, y_1), ...., (x_m, y_n)]$ for an $n \times m$ world, each item is assigned a value of zero for a non-traversable cell and one for a traversable cell. For the application of Anya, source and target points are selected uniformly from the set of traversable cells, and Anya is applied on the world map directly, considering each cell to be a uniform square of unit length. The resulting route is mapped to Spherical coordinates with the linear mapping, in degrees, being

\begin{align}\label{eq:spherical_mapping}
x_i^{lat} &= 90 \left (\frac{2x_i}{m} -1\right ) \nonumber\\
y_j^{lon} &= 180 \left (\frac{2y_j}{n} - 1\right ).
\end{align}

\begin{table}[t!]
\resizebox{\textwidth}{!}{%
\begin{tabular}{|l|l|l|l|l|l|l|l|l|l|l|l|l|l|l|l|l|l|}
\hline
\multicolumn{1}{|c|}{\multirow{2}{*}{Benchmark}} & \multicolumn{1}{c|}{\multirow{2}{*}{\#Maps}} & \multicolumn{1}{c|}{\multirow{2}{*}{\#Instances}} & \multicolumn{8}{c|}{Spherical Anya vs. Anya ET (All Routes)} \\ \cline{4-11} 
\multicolumn{1}{|c|}{} & \multicolumn{1}{c|}{} & \multicolumn{1}{c|}{} & \multicolumn{1}{c|}{\% less} & \multicolumn{1}{c|}{Min} & \multicolumn{1}{c|}{Q1} & \multicolumn{1}{c|}{Median} & \multicolumn{1}{c|}{Mean} & \multicolumn{1}{c|}{Q3} & \multicolumn{1}{c|}{Max} & \multicolumn{1}{c|}{StDev} \\ 
\hline
Starcraft & 75 & 18,068 & 35.27 & 0.003 & 0.806 & 1.188 & 2.048 & 2.018 & 60.019 & 3.451 \\
Dragon Age 2 & 67 & 14,369 & 31.58 & 0.006 & 0.857 & 1.000 & 1.368 & 1.525 & 56.366 & 1.359 \\
Baldurs Gate II, scaled to $512\times512$ & 75 & 17,580 & 31.18 & 0.005 & 0.861 & 1.000 & 1.561 & 1.690 & 53.262 & 1.798 \\
Warcraft III, scaled to $512\times512$ & 36 & 7,992 & 38.63 & 0.015 & 0.833 & 1.000 & 2.881 & 1.367 & 60.421 & 6.346 \\
City/street Maps & 90 & 21,515 & 33.53 & 0.014 & 0.804 & 1.261 & 3.424 & 3.103 & 60.442 & 6.160 \\
Random 10\%, $256\times512$ & 75 & 18,241 & 53.37 & 0.000 & 0.476 & 0.911 & 1.510 & 1.686 & 13.807 & 1.906 \\
Random 20\%, $256\times512$ & 75 & 18,305 & 46.48 & 0.000 & 0.593 & 1.052 & 1.587 & 1.876 & 13.782 & 1.724 \\
Random 30\%, $256\times512$ & 75 & 18,175 & 41.67 & 0.001 & 0.725 & 1.112 & 1.419 & 1.757 & 13.724 & 1.209 \\
Random 40\%, $256\times512$ & 75 & 16,091 & 31.73 & 0.000 & 0.944 & 1.100 & 1.207 & 1.359 & 13.582 & 0.694 \\
NOAA bathymetry, $10800\times21600$ & 1 & 22,060 & 74.20 & 0.003 & 0.796 & 0.905 & 1.630 & 1.000 & 90.596 & 5.085 \\
\hline
\end{tabular}
}
\caption{Execution time (ET) on a series of standard game maps, random maps and NOAA bathymetry is calculated from randomly sampled source and target points and alternatively considered to be defined over Euclidean and Spherical geometry. The column ``\% less" indicates the percentage of routes whereupon Spherical Anya ET is less than Anya ET.}
\label{table:results_table_et}
\end{table}

\begin{table}[t!]
\resizebox{\textwidth}{!}{%
\begin{tabular}{|l|l|l|l|l|l|l|l|l|l|l|l|l|l|l|l|l|l|}
\hline
\multicolumn{1}{|c|}{\multirow{2}{*}{Benchmark}} & \multicolumn{1}{c|}{\multirow{2}{*}{\#Maps}} & \multicolumn{1}{c|}{\multirow{2}{*}{\#Instances}} & \multicolumn{8}{c|}{Spherical Anya vs. Anya Route Length (Recipe 1)} \\ \cline{4-11} 
\multicolumn{1}{|c|}{} & \multicolumn{1}{c|}{} & \multicolumn{1}{c|}{} & \multicolumn{1}{c|}{\% less} & \multicolumn{1}{c|}{Min} & \multicolumn{1}{c|}{Q1} & \multicolumn{1}{c|}{Median} & \multicolumn{1}{c|}{Mean} & \multicolumn{1}{c|}{Q3} & \multicolumn{1}{c|}{Max} & \multicolumn{1}{c|}{StDev} \\ 
\hline
Starcraft & 75 & 18,068 & 56.98 & 0.132 & 0.990 & 1.000 & 0.977 & 1.000 & 1.509 & 0.069 \\
Dragon Age 2 & 67 & 14,369 & 36.27 & 0.162 & 1.000 & 1.000 & 0.995 & 1.000 & 2.047 & 0.047 \\
Baldurs Gate II, scaled to $512\times512$ & 75 & 17,580 & 36.71 & 0.191 & 0.999 & 1.000 & 0.997 & 1.000 & 5.528 & 0.082 \\
Warcraft III, scaled to $512\times512$ & 36 & 7,992 & 29.24 & 0.279 & 0.998 & 1.000 & 1.008 & 1.000 & 5.201 & 0.143 \\
City/street Maps & 90 & 21,515 & 60.54 & 0.180 & 0.910 & 0.996 & 0.947 & 1.000 & 5.983 & 0.137 \\
Random 10\%, $256\times512$ & 75 & 18,241 & 96.75 & 0.049 & 0.738 & 0.931 & 0.834 & 0.988 & 1.003 & 0.206 \\
Random 20\%, $256\times512$ & 75 & 18,305 & 95.96 & 0.046 & 0.748 & 0.932 & 0.836 & 0.988 & 1.002 & 0.205 \\
Random 30\%, $256\times512$ & 75 & 18,175 & 93.84 & 0.053 & 0.792 & 0.949 & 0.851 & 0.994 & 1.001 & 0.204 \\
Random 40\%, $256\times512$ & 75 & 16,091 & 84.41 & 0.024 & 0.965 & 0.998 & 0.924 & 1.000 & 1.001 & 0.172 \\
NOAA bathymetry, $10800\times21600$ & 1 & 22,060 & 37.44 & 0.320 & 1.000 & 1.000 & 0.997 & 1.000 & 1.103 & 0.017 \\
\hline
\end{tabular}
}
\caption{Route Length (RL) on a series of standard game maps, random maps and NOAA bathymetry is calculated from randomly sampled source and target points and alternatively considered to be defined over Euclidean and Spherical geometry. The table includes both illegal and legal Anya routes generated by interpolating between route points with a great-circle (Recipe 1). The column ``\% less" indicates the percentage of routes whereupon Spherical Anya RL is less than Anya RL. The ratio being greater than 1 occurs when non-legal routes are permitted, i.e. when Anya generates a route which in Spherical geometry intersects at least one non-traversable cell. All routes generated by Spherical Anya are legal.}
\label{table:results_table_recipe_1}
\end{table}

\begin{table}[]
\resizebox{\textwidth}{!}{%
\begin{tabular}{|l|l|l|l|l|l|l|l|l|l|l|l|l|l|l|l|l|l|}
\hline
\multicolumn{1}{|c|}{\multirow{2}{*}{Benchmark}} & \multicolumn{1}{c|}{\multirow{2}{*}{\#Maps}} & \multicolumn{1}{c|}{\multirow{2}{*}{\#Instances}} & \multicolumn{8}{c|}{Spherical Anya vs. Anya Route Length (Recipe 2)} \\ \cline{4-11} 
\multicolumn{1}{|c|}{} & \multicolumn{1}{c|}{} & \multicolumn{1}{c|}{} & \multicolumn{1}{c|}{\% less} & \multicolumn{1}{c|}{Min} & \multicolumn{1}{c|}{Q1} & \multicolumn{1}{c|}{Median} & \multicolumn{1}{c|}{Mean} & \multicolumn{1}{c|}{Q3} & \multicolumn{1}{c|}{Max} & \multicolumn{1}{c|}{StDev} \\ 
\hline
Starcraft & 75 & 18,068 & 91.22 & 0.131 & 0.969 & 0.993 & 0.965 & 0.999 & 1.000 & 0.072 \\
Dragon Age 2 & 67 & 14,369 & 86.90 & 0.162 & 0.996 & 0.999 & 0.990 & 1.000 & 1.000 & 0.043 \\
Baldurs Gate II, scaled to $512\times512$ & 75 & 17,580 & 75.04 & 0.191 & 0.989 & 0.999 & 0.984 & 1.000 & 1.000 & 0.045 \\
Warcraft III, scaled to $512\times512$ & 36 & 7,992 & 53.32 & 0.279 & 0.938 & 1.000 & 0.954 & 1.000 & 1.000 & 0.079 \\
City/street Maps & 90 & 21,515 & 85.94 & 0.176 & 0.870 & 0.971 & 0.917 & 0.999 & 1.000 & 0.111 \\
Random 10\%, $256\times512$ & 75 & 18,241 & 99.27 & 0.049 & 0.737 & 0.930 & 0.833 & 0.988 & 1.000 & 0.206 \\
Random 20\%, $256\times512$ & 75 & 18,305 & 99.62 & 0.046 & 0.748 & 0.931 & 0.836 & 0.988 & 1.000 & 0.205 \\
Random 30\%, $256\times512$ & 75 & 18,175 & 99.86 & 0.053 & 0.792 & 0.948 & 0.851 & 0.994 & 1.000 & 0.204 \\
Random 40\%, $256\times512$ & 75 & 16,091 & 99.93 & 0.024 & 0.965 & 0.998 & 0.924 & 1.000 & 1.000 & 0.172 \\
NOAA bathymetry, $10800\times21600$ & 1 & 22,060 & 49.79 & 0.320 & 0.999 & 1.000 & 0.993 & 1.000 & 1.000 & 0.024 \\
\hline
\end{tabular}
}
\caption{Distribution of Route Length Ratios (including illegal routes) where Anya routes are generated by linear enrichment every 1 arc-second between root point pairs for which a great-circle connection would yield an intersection in Spherical geometry (Recipe 2). Compared to the mechanism of Table \ref{table:results_table_recipe_1} this simultaneously reduces the number of illegal routes generated by Anya and increases the corresponding route lengths further from those of Recipe 1. The column ``\% less" indicates the percentage of routes whereupon Spherical Anya route length is less than Anya route length. The opposite occurs for illegal Anya routes.}
\label{table:results_table_recipe_2}
\end{table}
\begin{table}[]
\resizebox{\textwidth}{!}{%
\begin{tabular}{|l|l|l|l|l|l|l|l|l|l|l|l|l|l|l|l|l|}
\hline
\multicolumn{1}{|c|}{\multirow{2}{*}{Benchmark}} & \multicolumn{1}{c|}{\multirow{2}{*}{\#Instances}} & \multicolumn{8}{c|}{Spherical Anya vs. Anya Route Length (Exclusive, Recipe 1)} \\ \cline{3-10} 
\multicolumn{1}{|c|}{} & \multicolumn{1}{c|}{} & \multicolumn{1}{c|}{\% less} & \multicolumn{1}{c|}{Min} & \multicolumn{1}{c|}{Q1} & \multicolumn{1}{c|}{Median} & \multicolumn{1}{c|}{Mean} & \multicolumn{1}{c|}{Q3} & \multicolumn{1}{c|}{Max} & \multicolumn{1}{c|}{StDev} \\ 
\hline
Starcraft & 12,594 & 81.75 & 0.132 & 0.970 & 0.997 & 0.964 & 1.000 & 1.000 & 0.077 \\
Dragon Age 2 & 9,511 & 54.80 & 0.162 & 0.999 & 1.000 & 0.990 & 1.000 & 1.000 & 0.049 \\
Baldurs Gate II, scaled to $512\times512$ & 12,454 & 51.81 & 0.191 & 0.996 & 1.000 & 0.987 & 1.000 & 1.000 & 0.045 \\
Warcraft III, scaled to $512\times512$ & 6,295 & 37.12 & 0.279 & 0.989 & 1.000 & 0.978 & 1.000 & 1.000 & 0.052 \\
City/street Maps & 17,094 & 76.20 & 0.180 & 0.875 & 0.974 & 0.922 & 1.000 & 1.000 & 0.107 \\
Random 10\%, $256\times512$ & 17,997 & 98.07 & 0.049 & 0.734 & 0.929 & 0.832 & 0.987 & 1.000 & 0.206 \\
Random 20\%, $256\times512$ & 17,860 & 98.35 & 0.046 & 0.740 & 0.927 & 0.832 & 0.986 & 1.000 & 0.206 \\
Random 30\%, $256\times512$ & 17,289 & 98.65 & 0.053 & 0.777 & 0.941 & 0.844 & 0.990 & 1.000 & 0.206 \\
Random 40\%, $256\times512$ & 13,716 & 99.02 & 0.024 & 0.945 & 0.995 & 0.911 & 1.000 & 1.000 & 0.183 \\
NOAA bathymetry, $10800\times21600$ & 19,924 & 41.46 & 0.320 & 1.000 & 1.000 & 0.996 & 1.000 & 1.000 & 0.017 \\
\hline
\end{tabular}
}
\caption{Distribution of Route Length Ratios excluding illegal routes wherein the Anya route is shorter than the Spherical Anya route. This occurs if and only if Anya returns a route which intersects a non-traversable cell at least once.}
\label{table:results_legal}
\end{table}

\noindent Results are shown for two distinct recipes for connecting the root points returned from Anya to form a route in Spherical geometry. \textit{Recipe 1} connects root points directly with great-circle lines. A consequence of this is that Recipe 1 achieves routes whose total length is closer to that of the optimal route returned by Spherical Anya, but with a relatively high number of routes intersecting non-traversable cells. We refer to routes as \textit{legal} whereupon this does not occur. All routes generated by Spherical Anya in Spherical geometry are both legal and optimal whereas this is not the case for Anya -- Anya not being designed for application in Spherical geometry and therefore guaranteeing neither legality nor optimality of routes. \textit{Recipe 2} enriches the route segments returned by Recipe 1 which intersect non-traversable cells with additional synthetic root points every arc-second along a straight (Euclidean) line between the root points where this occurs. Finally a great-circle line is assumed between root-point pairs in the enriched set. Recipe 2 produces fewer illegal routes than Recipe 1, albeit at the expense of generating legal routes with distance greater than those produced by Recipe 1.

For Spherical Anya the world map is first transformed into Spherical coordinates and the shortest-distance route is subsequently calculated. Under such circumstances it is guaranteed that Spherical Anya will find the optimal legal route whereas Anya may not, the relevant question here being \textit{at what cost}? What is the cost in execution time (ET) and route length (RL) of achieving optimality in Spherical geometry, compared with the simplification of assuming a Euclidean geodesic on a ``flat earth''.

Results are provided in Tables \ref{table:results_table_et}, \ref{table:results_table_recipe_1}, \ref{table:results_table_recipe_2} and \ref{table:results_legal}. The standard benchmark game maps which have been selected are \textit{Starcraft}, \textit{Dragon Age 2}, \textit{Baldurs Gate II} and \textit{Warcraft III}, with source and target uniformly selected with replacement an \textit{\#Instances} number of times over a \textit{\#Maps} number of maps for each game. These game maps have features which are interesting to study in their own right and we do not dwell on them here. For example Real Time Strategy games such as \textit{Starcraft} contain large, obstacle free, traversable areas with choke points (see Figure \ref{fig:sc_maps}) whereas Role-playing games such as \textit{Baldurs Gate II} are maze-like. Figure \ref{fig:benchmark_et_game} illustrates the ratio of Spherical Anya ET over Anya ET for these game maps and \textit{City/street} maps. Also presented are results from random maps (see Figure \ref{fig:benchmark_et_random}). Note that these maps are originally 2D maps, repurposed here for illustrating the statistical relationship between Anya and Spherical Anya given (i) different recipes for creating routes in Spherical geometry from the root points yielded by Anya (ii) different underlying measures for line-of-sight checking and finally, but significantly (iii) differing world maps. Beyond this statistical comparison it is the sea-voyages based on the NOAA bathymetric relief map of Earth \cite{noaa} which additionally bear practical utility in real world applications.

These simulations show that Spherical Anya is mostly similar in execution time to Anya, but on average it is slower (see Table \ref{table:results_table_et}). Moreover the median ET of Spherical Anya being faster or slower than Anya is strongly dependent on the features of each world map. The reward is only occasionally speed but always legality and optimality, Spherical Anya producing routes on average over 7.4\% shorter than routes generated by Anya across all maps using Recipe 1 (see Table \ref{table:results_legal}). The statistics for Random maps at $10\%$, $20\%$, $30\%$, $40\%$ clearly show the importance of the features of the world map as well as the number of tiles crossed (see Figure \ref{fig:benchmark_et_random}), the percentage here defining the Bernoulli random variable determining the existence of a non-traversable coordinate in the world map. Random maps produce route length ratios smaller than those of other maps and suffer a high RL ratio standard deviation compared to other world maps. More importantly for practical applications such as the navigation of seafaring vessels, median ET for Spherical Anya is slightly less than Anya and average ET is higher for almost all world map resolutions and route lengths (see Figures \ref{fig:benchmark_candlesticks} and \ref{fig:benchmark_length_candlesticks}). Never generating routes which intersect land, as Anya does in Spherical geometry, Spherical Anya always returns the optimal route. Refer to Appendix B, Figures \ref{fig:ex_1} and \ref{fig:ex_2} for illustrations of instances of ocean voyages. In summary, Spherical Anya being always optimal and producing legal routes in Spherical geometry, is faster on most routes compared with Anya for sea-faring routes, but more often it is slower.

\begin{figure*}[]
    \centering
    \begin{subfigure}[t]{0.5\textwidth}
        \centering
\includegraphics[height=1\linewidth]{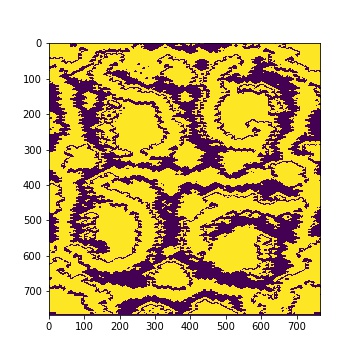}
        \caption{A Starcraft map instance, with traversable regions coloured yellow and non-traversable regions coloured purple. This can be considered either a flat map, as was originally intended, or a map in Spherical geometry via the mapping of equation  (\ref{eq:spherical_mapping}).}
    \end{subfigure}%
    ~~~~~~~~ 
    \begin{subfigure}[t]{0.5\textwidth}
        \centering
        \includegraphics[height=1\linewidth]{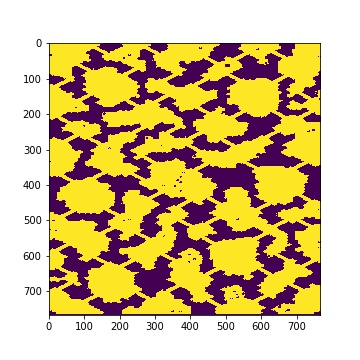}
        \caption{A second example of a Starcraft map instance. Here the ratio of non-traversable regions compared to traversable regions is 38.37\%.}
    \end{subfigure}
    
     \begin{subfigure}[t]{0.5\textwidth}
        \centering
        \includegraphics[height=.9\linewidth]{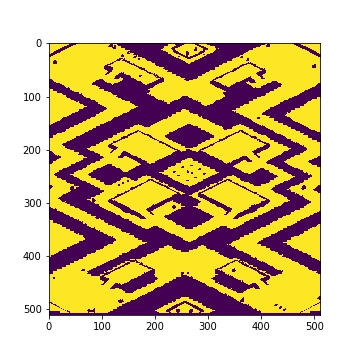}
        \caption{A third example of a Starcraft map instance. The ratio of non-traversable cells compared to traversable cells is 42.11\%. The average length of a route for this map is 10,920.2 km in Euclidean geometry and 10,829.7 km in Spherical geometry. To compare, the median distance over all Starcraft map instances in Euclidean and Spherical geometry is 13,822.4 km and 13,454.9 km respectively.}
    \end{subfigure}%
    ~~~~~~~~ 
    \begin{subfigure}[t]{0.5\textwidth}
        \centering
                \includegraphics[height=.9\linewidth]{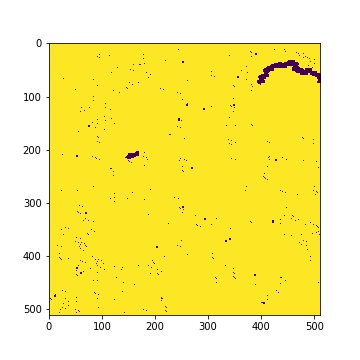}
        \caption{A final example of a Starcraft map instance. The ratio of non-traversable cells compared to traversable cells is 1.38\%. The average length of a route for this map is 10,606.3 km in Euclidean geometry and 9,577.0 km in Spherical geometry. The median ratio over all Starcraft map instances is 25.77\%.}
    \end{subfigure}
    \caption{Starcraft Real Time Strategy game map instance examples.}
    \label{fig:sc_maps}
\end{figure*}

\begin{figure*}[t!]
\centering
\includegraphics[width=\linewidth]{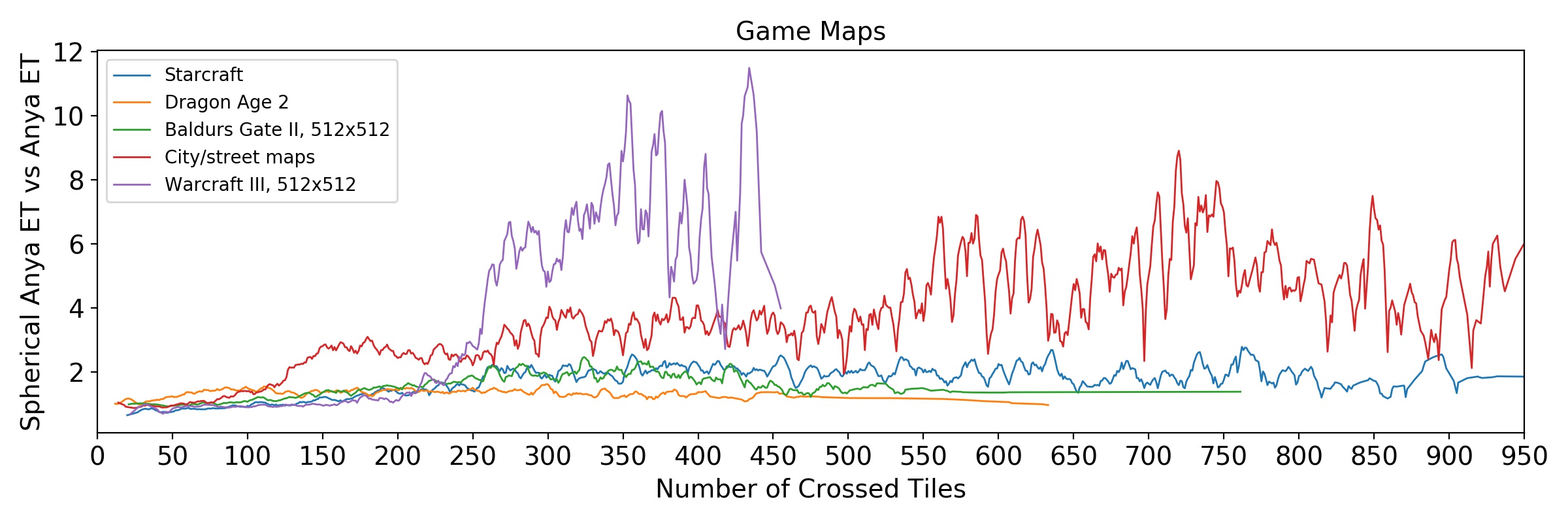}
        \caption{Spherical Anya execution time (ET) over Anya ET against Number of Tiles Crossed (CT) for Game and City/street maps.}
        \label{fig:benchmark_et_game}
\end{figure*}

\begin{figure*}[t!]
\centering
\includegraphics[width=\linewidth]{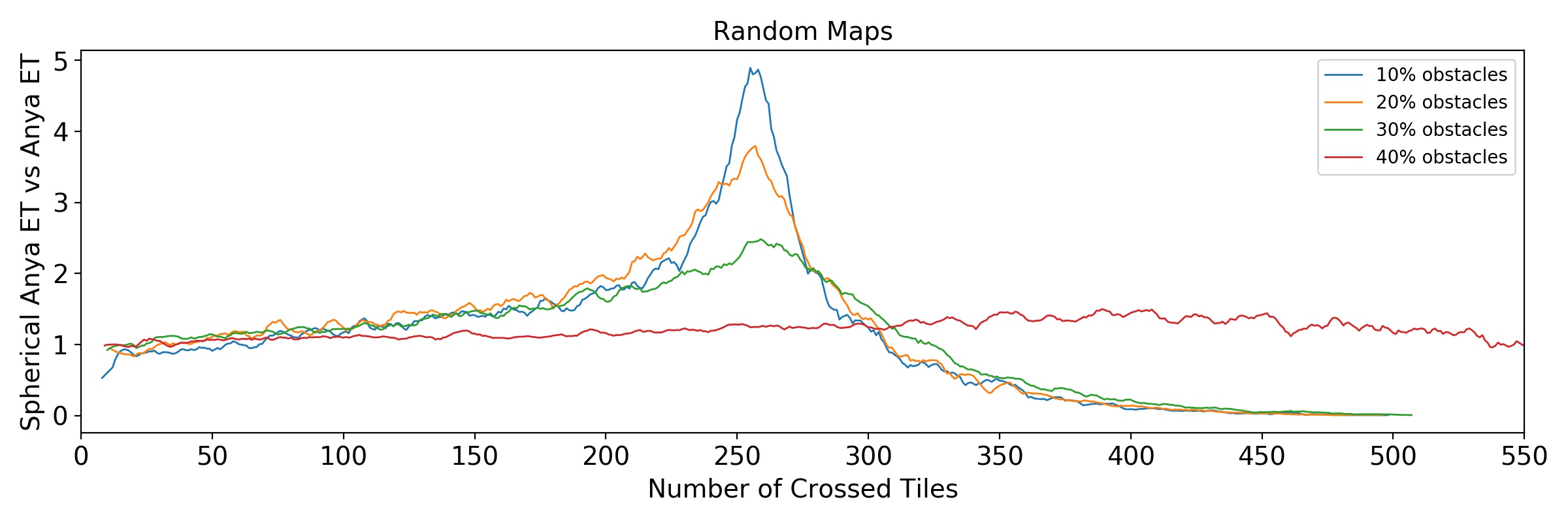}
        \caption{Spherical Anya execution time (ET) over Anya ET against Number of Tiles Crossed (CT) for Random Maps.}
    \label{fig:benchmark_et_random}
\end{figure*}

\begin{figure*}[t!]
\centering
\includegraphics[width=.6\linewidth]{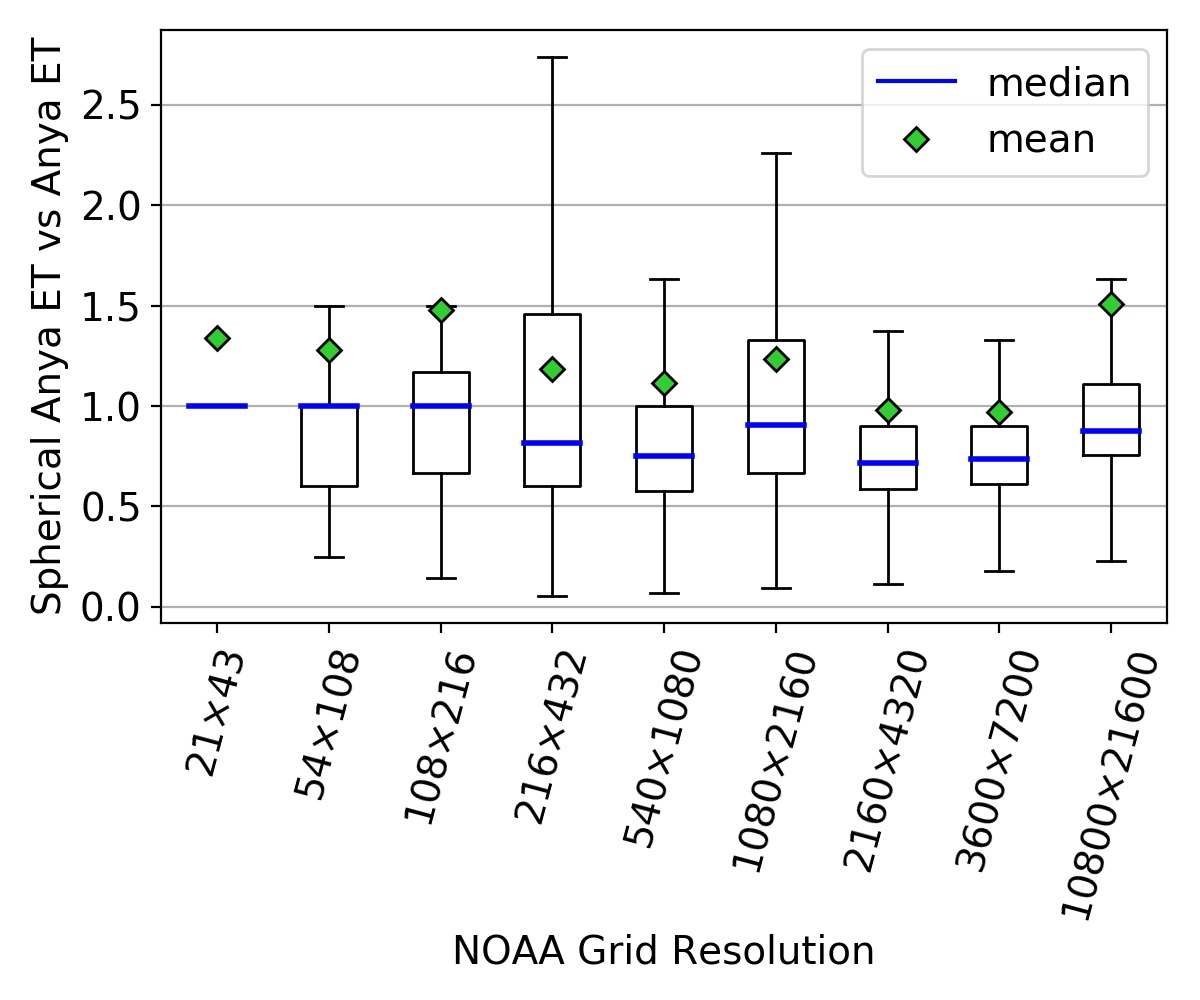}
        \caption{Spherical Anya execution time (ET) over Anya ET for Earth NOAA bathymetric maps taken at several resolutions. This indicates that for most routes one can calculate optimal voyages for seafaring vessels \textit{faster} than with a ``flat earth" approximation, albeit the execution time is on average slower.}
    \label{fig:benchmark_candlesticks}
\end{figure*}

\begin{figure*}[t!]
\centering
\includegraphics[width=1.1\linewidth]{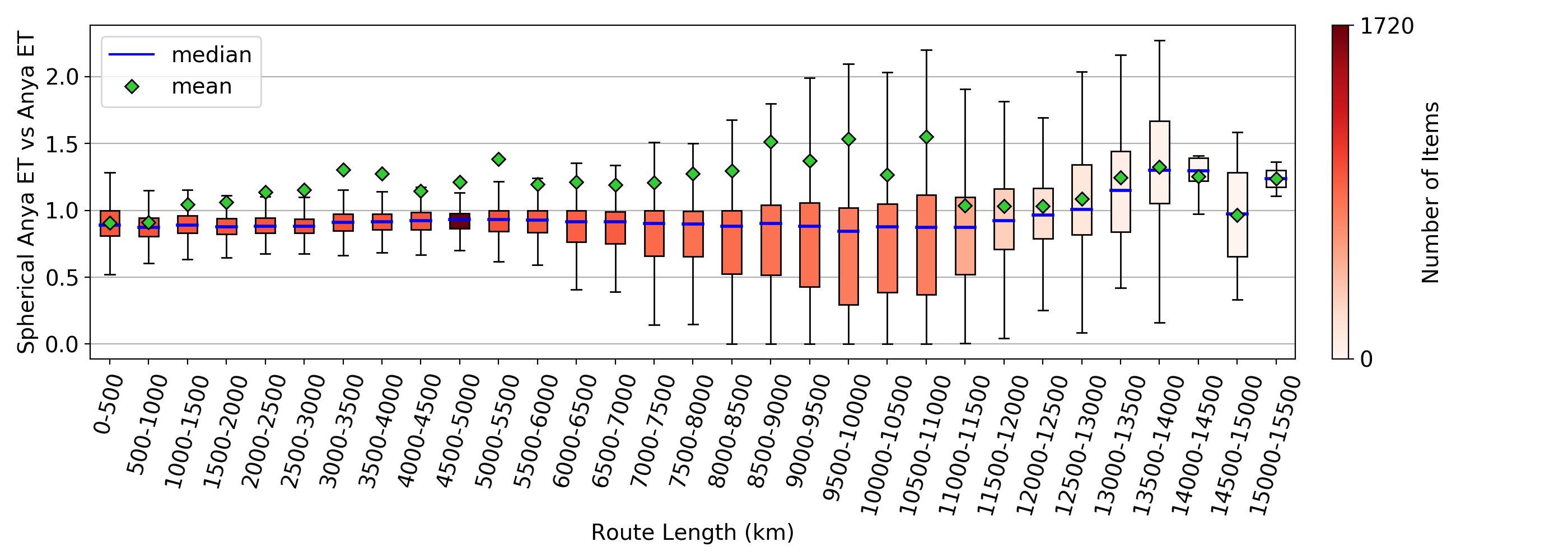}
        \caption{Spherical Anya execution time (ET) over Anya ET for buckets of route lengths at intervals of 500 km using the NOAA Earth bathymetric map at resolution $10,800\times21,600$.}
    \label{fig:benchmark_length_candlesticks}
\end{figure*}

\section{Conclusion}
It has been shown that a fast, efficient and optimal pathfinding algorithm for Spherical geometry exists in the form of Spherical Anya, a new online pathfinding algorithm which modifies and extends the principles of the Anya algorithm in Euclidean geometry. Modifications include the obvious change from the Euclidean geodesic to the great-circle geodesic of Spherical geometry. Further changes are a consequence of this, for instance one can no longer rely upon all turning points on an optimal route being corner points, and the existence of adjoint search nodes and successors which do not exist in Anya. The Spherical Anya algorithm for optimal pathfinding in Spherical geometry is described and illustrated with examples, and the primary mathematical consequence of the transition from Euclidean geometry to Spherical geometry is stated and proven. All remaining terminology, procedures and fundamental principles remaining the same as Anya ensures that Spherical Anya adopts the completeness, optimality and termination guarantees of its progenitor.

Finally, examples are shown applying Anya and Spherical Anya to several game maps, random maps and a map of Earth formed of bathymetric data from NOAA. In each case the map is assumed to be a representation of a world in Spherical geometry, and the application of Anya therefore a first order approximation of the optimal route, whereby a ``flat earth" approximation is made for all projections, but the resulting series of root points are mapped to Spherical geometry, route length being calculated as the sum of great-circle arcs between them (Recipe 1) or a sequence of enriched root point pairs joined by great-circle lines (Recipe 2). These simulations illustrate the following primary features and advantages of Spherical Anya. Firstly what is known mathematically in advance; that the result from Spherical Anya is always optimal in Spherical geometry and that Anya is not (and is not intended to be). Secondly, that applying Anya naively to a spherical world map to generate a sequence of root point pairs joined by great-circle lines yields a large proportion of illegal routes which at times cross non-traversable cells. The quantity of illegal routes is a function of the recipe used to convert root points in Euclidean geometry to Spherical geometry, however it is also important to note that as of yet there is no established method which can guarantee all routes produced by Anya are legal in Spherical geometry. And lastly, that Spherical Anya is more often faster than Anya for sea-voyages but in general and on average it is slower, the relative statistical properties depending not only on the algorithm but on the features of each world map.

\appendix
\section*{Appendix A. Gnomonic Projection}
A Gnomonic projection is a mapping of a coordinate system wherein all great-circle lines are represented by a straight line and vice versa (see Figure \ref{fig:gnomonic_projection}). Used widely in Astronomy, Seismology, Signals Intelligence, air traffic and maritime navigation, a Gnomonic projection cannot be used to enable an error-free application of Anya directly. The reason is understood with the following considerations.

The conversion of world map points defined in spherical coordinates with a Gnomonic projection centred on the equator $(0,0)$ is given by 
\begin{align}
	x &= \tan(\theta)\notag\\
	y &= \frac{\tan(\phi)}{\cos(\theta)} 
	\label{equation:gnomonic}
\end{align}

\begin{figure*}[h!]
    \centering
    \begin{subfigure}[t]{0.49\textwidth}
        \centering
\includegraphics[width=\linewidth]{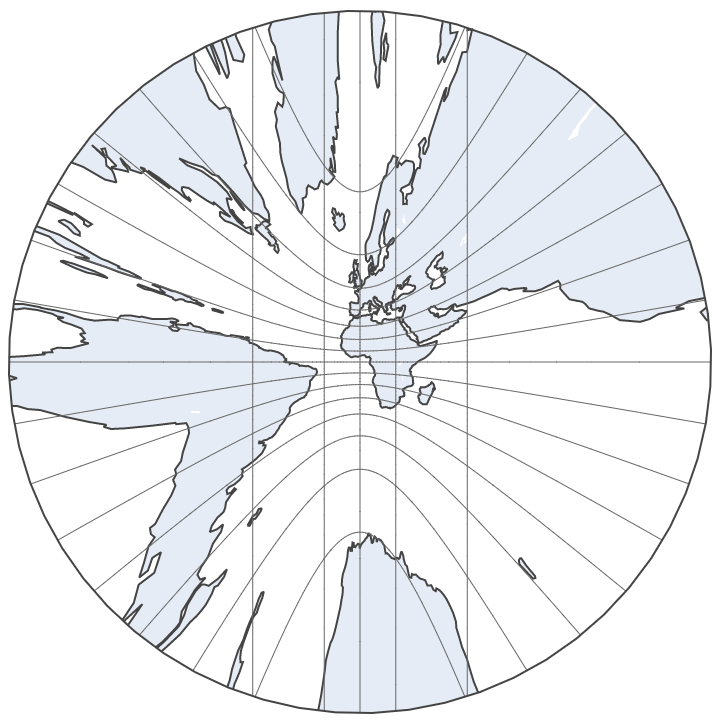}
        \caption{Gnomonic projection of Earth hemisphere centred on $(0,0)$. Longitudinal lines, being great circle routes, manifest as straight vertical lines in this figure.}
    \end{subfigure}%
    ~~~~~~~~
    \begin{subfigure}[t]{0.49\textwidth}
        \centering
\includegraphics[width=\linewidth]{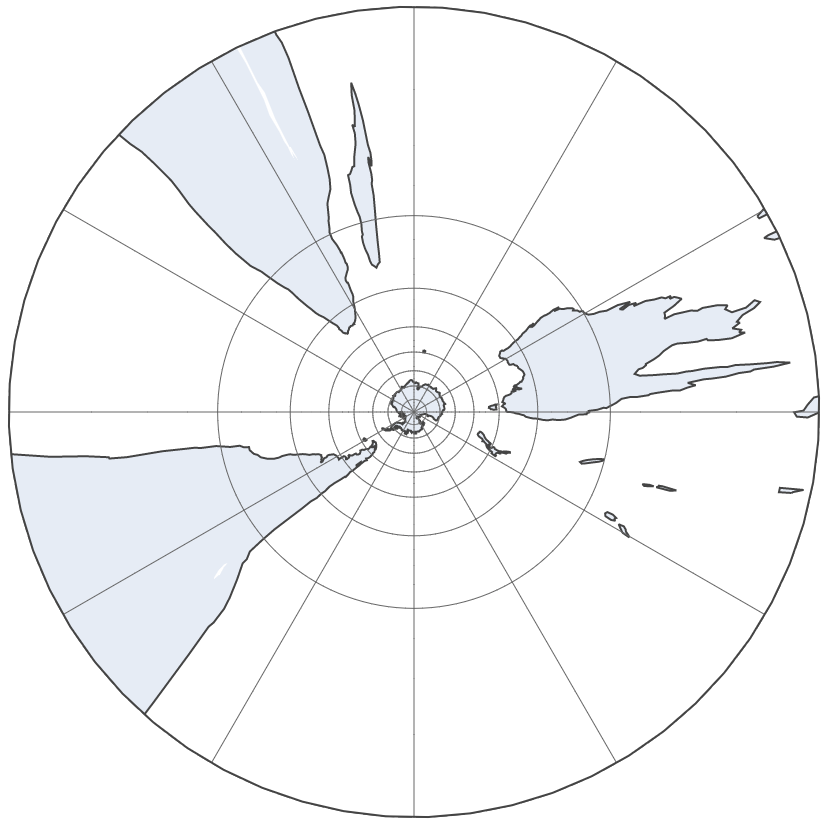}
        \caption{Hemisphere centred on Antarctic with Gnomonic projection centred on $(0,0)$. The hemisphere view being centred on the Antarctic in this case, great circle lines of longitude are now seen to be straight -- radially directed -- lines.}
    \end{subfigure}
    \caption{On a Gnomonic projection of a spherical world map all straight lines describe a great circle. This would be a suitable reference frame for pathfinding on a sphere, but for the fact errors in estimating obstacle boundaries cannot be removed in this case.}
    \label{fig:gnomonic_projection}
\end{figure*}

\noindent where $\theta$ is the azimuthal angle and $\phi$ is the polar angle. Hence latitudinal and longitudinal values from each hemisphere will map to the same $x$ and $y$ values, necessitating that routes in each hemisphere are treated with distinct coordinate systems. The dependence of $y$ upon both $\phi$ and $\theta$ implies that points on a line in Spherical geometry are not uniformly divided and an expanded basis set is required to define every $y$-axis with respect to a given $x$ (see Figure \ref {fig:grid_transformation}). In practice, this means that any visibility graph defined on the Gnomonic world map requires not an $M \times N$ as is the case in Euclidean and Spherical geometry, but instead an $M \times N^2$ matrix -- potentially in excess of available memory for a large number for real-world applications. More problematic is that implementing an error-free application of Anya using a Gnomonic projection is mathematically impossible as can be shown using the two non-traversable cells illustrated in Figure \ref{fig:grid_transformation}. The square corners of obstacle $A$, transformed by equation \ref{equation:gnomonic}, form obstacle $A'$ in Gnomonic space and analogously $B$ will transform to $B'$. Lines of constant longitude in Spherical geometry will remain straight horizontal lines under the Gnomonic projection, whereas lines of constant latitude will be curved. The corresponding edges of $A$ compared to $B$, equal in Spherical geometry in units of $\phi$ and $\theta$, will no longer be equal under the Gnomonic projection. The right and left edges of each obstacle, being defined by a line of constant latitude, form curved lines under the Gnomonic projection, therefore intersecting the lines of the grid regardless of the world map resolution. In other words, there will always be an error in the approximated boundary of the transformed obstacles and this error, being a function of $\theta$ and $\phi$, will be non-uniform. We should note that this statement applies exclusively to the error in representing mathematically ideal (perfectly accurate) obstacles using a Gnomonic projection and that it is specifically this which prohibits a guarantee of optimality being associated with the Gnomonic approach. However in the absence of ideal obstacles, such as representations of real-world objects which are prone to measurement error, both the Gnomonic approach and the approach of Spherical Anya as outlined in this paper may suffer sub-optimality.

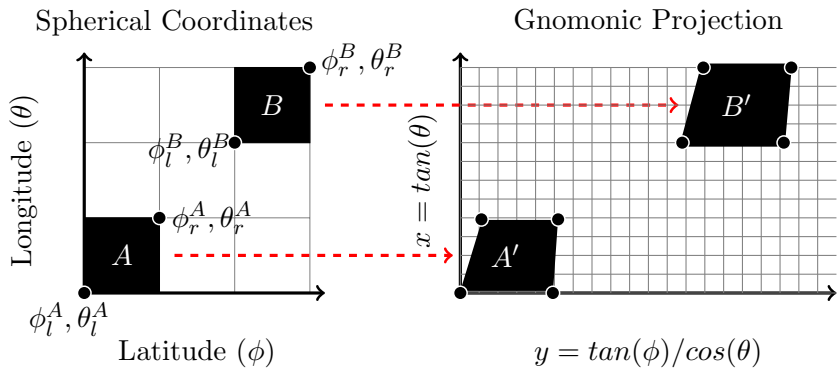
\begin{figure*}[]
\centering
\begin{tikzpicture}[scale=1]
\draw[step=1cm,gray,very thin] (0,0) grid (3,3);
\filldraw[fill=black, draw=black] (0,0) rectangle (1,1);
\filldraw[fill=black, draw=black] (2,2) rectangle (3,3);
\node[white] at (0.5, 0.5) {$A$};
\node[white] at (2.5, 2.5) {$B$};
\node[black] at (1.2,3.6) {Spherical Coordinates};
\node[black] at (7.5,3.6) {Gnomonic Projection};
\node[black, rotate=90] at (-.8, 1.5) {Longitude ($\theta$)};
\node[black] at (1.5,-.8) {Latitude ($\phi$)};
\draw[->, very thick] (5, 0) -- (5, 3.2);
\draw[->, very thick] (5, 0) -- (10, 0);
\draw[->, very thick] (0, 0) -- (0, 3.2);
\draw[->, very thick] (0, 0) -- (3.2, 0);

\draw[dashed, ->, red, very thick] (1.2, .5) -- (4.9, .5);
\draw[dashed, ->, red, very thick] (3.2, 2.5) -- (7.9, 2.5);

\foreach \i in {0,1,2,3,4,5,6,7,8,9,10,11,12}
	\draw[-, gray] (5, .25*\i) -- (10, .25*\i);
\foreach \z in {0,1,2,3,4,5,6,7,8,9,10,11,12,13,14,15,16,17,18,19}
	\draw[-, gray] (5+.25*\z, 0) -- (5+.25*\z, 3.);
	
	\node[trapezium, fill, minimum width=1.3cm, trapezium left angle=73, trapezium right angle=94] at (5.76,.5) {};
\node[trapezium, fill, minimum width=1.5cm, trapezium left angle=75, trapezium right angle=95] at (8.76,2.5) {};

\node[white] at (5.6, 0.5) {$A'$};
\node[white] at (8.67, 2.5) {$B'$};
\node[black, rotate=90] at (4.5, 1.5) {$x = tan(\theta)$};
\node[black] at (7.5,-.8) {$y = tan(\phi)/cos(\theta)$};

\node[black] at (-0.2, -0.35) {${\phi_{l}^A}, {\theta_{l}^A}$};
\node[black] at (1.7, 1.) {${\phi_{r}^A}, {\theta_{r}^A}$};
\node[black] at (1.4, 1.9) {${\phi_{l}^B}, {\theta_{l}^B}$};
\node[black] at (3.7, 3.1) {${\phi_{r}^B}, {\theta_{r}^B}$};

\fill[white] (0, 0) circle[radius=1mm];
\fill[white] (1, 1) circle[radius=1mm];
\fill[white] (2, 2) circle[radius=1mm];
\fill[white] (3, 3) circle[radius=1mm];
\fill[white] (5, 0) circle[radius=1mm];
\fill[white] (6.23, 0) circle[radius=1mm];

\fill[white] (5.28, .98) circle[radius=1mm];
\fill[white] (6.3, .98) circle[radius=1mm];
\fill[white] (8.23, 3) circle[radius=1mm];
\fill[white] (9.4, 3) circle[radius=1mm];
\fill[white] (7.95, 2) circle[radius=1mm];
\fill[white] (9.3, 2) circle[radius=1mm];

\fill[black] (0, 0) circle[radius=.8mm];
\fill[black] (1, 1) circle[radius=.8mm];
\fill[black] (2, 2) circle[radius=.8mm];
\fill[black] (3, 3) circle[radius=.8mm];
\fill[black] (5, 0) circle[radius=.8mm];
\fill[black] (6.23, 0) circle[radius=.8mm];
\fill[black] (5.28, .98) circle[radius=.8mm];
\fill[black] (6.3, .98) circle[radius=.8mm];
\fill[black] (8.23, 3) circle[radius=.8mm];
\fill[black] (9.4, 3) circle[radius=.8mm];
\fill[black] (7.95, 2) circle[radius=.8mm];
\fill[black] (9.3, 2) circle[radius=.8mm];

\end{tikzpicture}
\caption{Two non-traversable cells or \textit{obstacles} $A$ and $B$ in Spherical geometry (left). Under a Gnomonic projection, lines of constant latitude will intersect with the grid ensuring that any grid, regardless of the interval length, will always result in an error in the approximated boundary of a transformed obstacle and this error, being a function of $\theta$ and $\phi$, will be non-uniform.}
\label{fig:grid_transformation}
\end{figure*}

\appendix
\section*{Appendix B. Examples of routes in NOAA bathymetric dataset}

Optimal pathfinding for ocean voyages is an important practical problem for the maritime industry and global trade, wherein commercial vessels must choose the most expedient routes in order to ensure fast and efficient delivery of cargo.

To measure the performance of Spherical Anya on a real-world problem, more than 20,000 route instances were calculated using the NOAA bathymetric dataset at resolution $10,800 \times 21,600$. Comparing against Anya the statistics are as per Tables \ref{table:results_table_et}, \ref{table:results_table_recipe_1}, \ref{table:results_table_recipe_2} and \ref{table:results_legal}. Table \ref{table:results_table_recipe_1} includes all routes, including non-legal routes and 47.2\% of these correspond to cases where the length of the Anya route (as per Recipe 1) is shorter than the length of the corresponding Spherical Anya route, which occurs due to the geometric feature wherein a straight line in Euclidean geometry which doesn't intersect any barriers, nonetheless crosses them in Spherical geometry. Figure \ref{fig:ex_1} illustrates an example of this for the shortest sea-passage between New York, USA and Gothenburg, Sweden where a great-circle connecting Anya root points intersects land in Spherical geometry. An example of the opposite case, illustrated in Figure \ref{fig:ex_2}, is of Spherical Anya and Anya routes connecting Bahia Blanca, Argentina and Colombo, Sri Lanka, where neither of the routes intersect land. When this happens, Spherical Anya always returns the shortest route or at best, the two routes coincide.

\begin{figure*}[t!]
\centering
\includegraphics[width=0.95\linewidth]{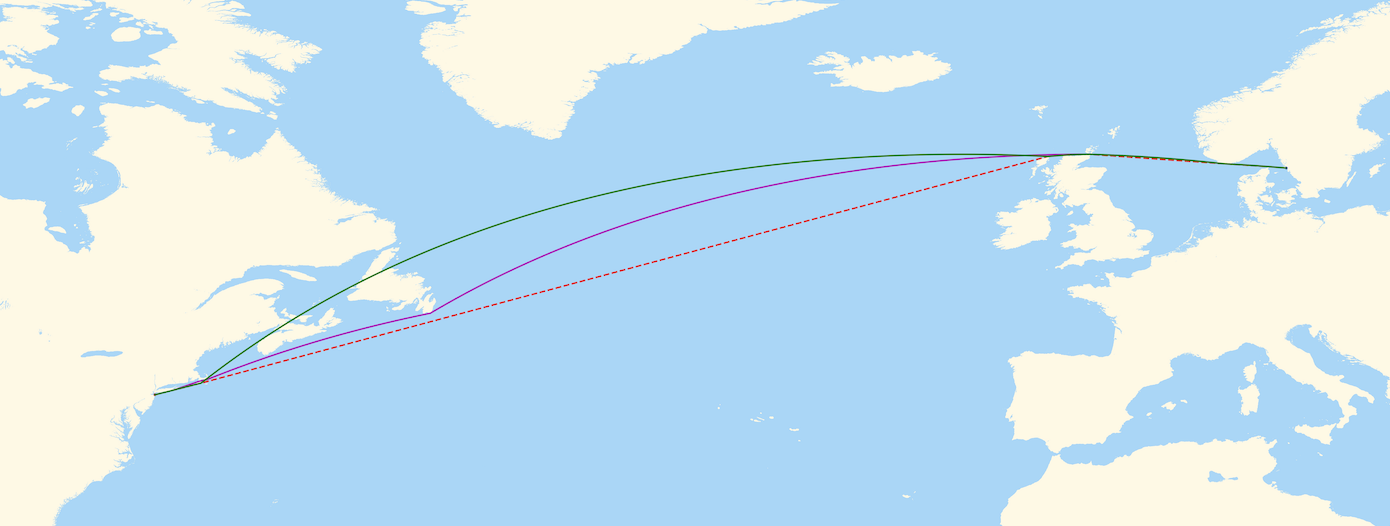}
        \caption{Spherical Anya route (solid purple line) and Anya route (dashed red line) from New York, USA to Gothenburg, Sweden. The sequence of great-circles connecting root points of the Anya route are shown with a green solid line. Some of these great-circles intersect land. The length of the Spherical Anya route is 6,222.888 km, while the total length of all green great-circles is 6,132.063 km. The Spherical Anya route is calculated 4.65 times faster than the Anya route.}
    \label{fig:ex_1}
\end{figure*}

\begin{figure*}[t!]
\centering
\includegraphics[width=0.95\linewidth]{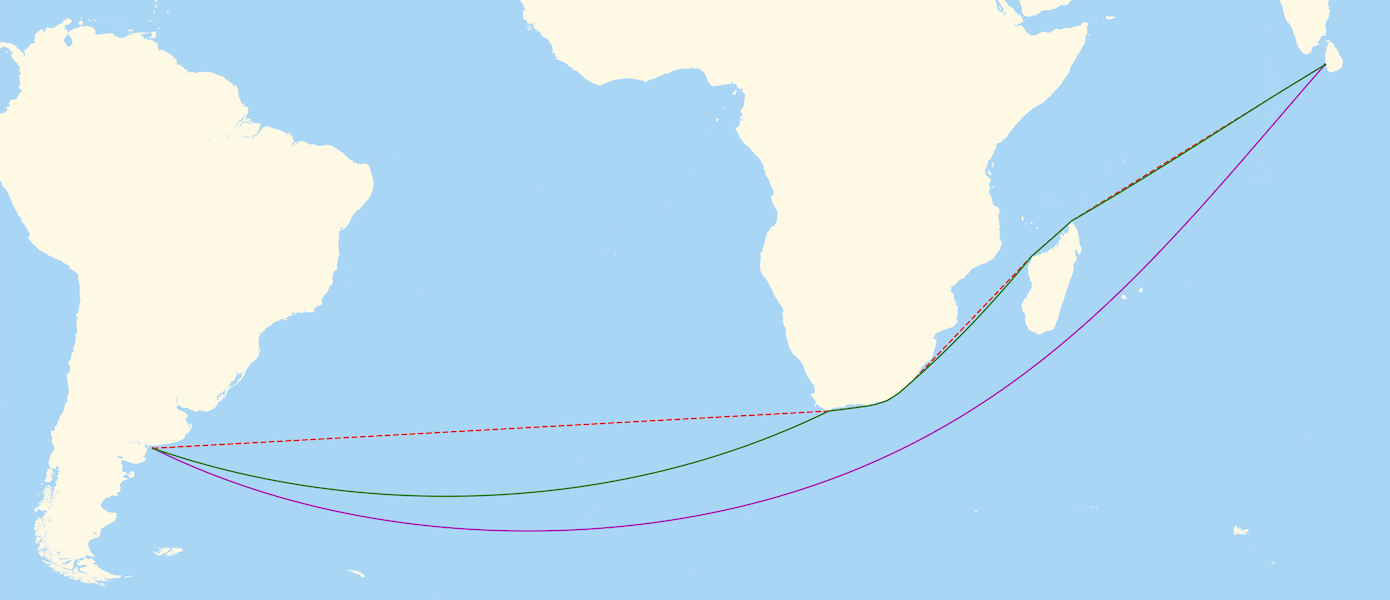}
        \caption{Spherical Anya route (solid purple line) and Anya route (dashed red line) from Bahia Blanca, Argentina to Colombo, Sri Lanka. The sequence of great-circles connecting root points of the Anya route are shown with a green solid line. None of these great-circles intersect land. The length of the Spherical Anya route is 14,715.343 km, while total length of all green great-circles is 14,918.216 km. The Spherical Anya route is calculated 311 times faster than the Anya route.}
    \label{fig:ex_2}
\end{figure*}

\FloatBarrier
\vskip 0.2in

\newpage
\bibliography{anya_in_spherical_geometry.bib}
\bibliographystyle{theapa}

\end{document}